\documentclass[a4paper,onecolumn,11pt,accepted=2024-06-02]{quantumarticle}
\pdfoutput=1
\usepackage[utf8]{inputenc}
\usepackage[english]{babel}
\usepackage[T1]{fontenc}
\usepackage{amsmath}
\usepackage{hyperref}

\usepackage{amssymb}
\usepackage{subcaption}
\usepackage{graphicx}
\usepackage{arydshln}

\usepackage{amsthm,scalerel}

\usepackage[matrix,frame,arrow]{xypic}
\usepackage[braket]{qcircuit}
\xyoption{color}
\newcommand{\qwblank}[1][-1]{\ar @[white]@{-} [0,#1]}
\newcommand{\ctrlblank}[1]{\qwx[#1] \qwblank}
\newcommand{\ctrlblankA}[1]{\control \qwx[#1] \qwblank}
\newcommand{\targblank}{*+<.02em,.02em>{\xy ="i","i"-<.39em,0em>;"i"+<.39em,0em> **\dir{-}, "i"-<0em,.39em>;"i"+<0em,.39em> **\dir{-},"i"*\xycircle<.4em>{} \endxy} \qwblank}

\newcommand{\etal}{\textit{et~al.}}
\newcommand{\ie}{\textit{i.e.}}
\newcommand{\etc}{etc.}

\newtheorem{theorem}{Theorem}[section]
\newtheorem{lemma}[theorem]{Lemma}
\newtheorem{corollary}{Corollary}[theorem]
\newtheorem{conjecture}{Conjecture}[section]

\newtheorem{definition}[theorem]{Definition}

\numberwithin{equation}{section}

\begin{document}

\title{On Groups in the Qubit Clifford Hierarchy}

\author{Jonas T. Anderson}
\address{}
\email{jonas.tyler.anderson@gmail.com}
\affiliation{Northrop Grumman Corporation}
\orcid{0000-0002-3230-6712}

\maketitle

\begin{abstract}

Here we study the unitary groups that can be constructed using elements from the qubit Clifford Hierarchy. We first provide a necessary and sufficient canonical form that semi-Clifford and generalized semi-Clifford elements must satisfy to be in the Clifford Hierarchy. Then we classify the groups that can be formed from such elements. Up to Clifford conjugation, we classify all such groups that can be constructed using generalized semi-Clifford elements in the Clifford Hierarchy. We discuss a possible minor exception to this classification in the appendix. This may not be a full classification of all groups in the qubit Clifford Hierarchy as it is not currently known if all elements in the Clifford Hierarchy must be generalized semi-Clifford. In addition to the diagonal gate groups found by Cui et al., we show that many non-isomorphic (to the diagonal gate groups) generalized symmetric groups are also contained in the Clifford Hierarchy. Finally, as an application of this classification, we examine restrictions on transversal gates given by the structure of the groups enumerated herein which may be of independent interest.
\end{abstract}

\section{Prologue: The Clifford Hierarchy}

The Pauli and Clifford groups are ubiquitous in quantum information processing. The Clifford Hierarchy, while less widely known, builds upon these groups (though it is not a group) and is frequently encountered in its own right. It typically rears its head in tasks centered on processing of stabilizers states (or groups) when processes involving Clifford operations are assumed to be free (or at least easy) to implement. The difficulty of the task typically increases with level in the Hierarchy. 

The Clifford Hierarchy ($\mathcal{CH}$) was originally defined \cite{Gottesman1999} as the set of gates that could be realized via gate teleportation with Pauli gate corrections though it has grown to encompass many areas of research. 

A (likely incomplete) list of myriad places where the Clifford Hierarchy is featured follows: In magic state distillation schemes \cite{Bravyi:2005a, Bravyi2012,Meier:2012a}, the state distilled is typically used as a resource state to apply a gate from the lower levels of the Clifford Hierarchy. Resource theories classifying the difficulty of preparing magic states find that higher levels in $\mathcal{CH}$ tend to correspond to more `valuable' resources \cite{Howard:2017a}. Transversal gates (and even low-depth circuits) on stabilizer codes have been shown \cite{OConnor:2018a} to only implement logical operators from finite levels in the $\mathcal{CH}$. Also, the local physical gates which can be applied to implement a transversal gate on a qubit stabilizer code have been shown \cite{Anderson:2016a,Englbrecht2020} in some cases to be in $\mathcal{CH}$. With few exceptions, local symmetries of stabilizer states and graph states are in the Clifford Hierarchy \cite{Englbrecht2020}. Many normal forms use unitary gate sets consisting of the Clifford group and some element(s) in the lower levels of the Clifford Hierarchy. These normal forms, in turn, provide the mathematical underpinning for exact (and approximate) gate compiling algorithms. Recently \cite{Frembs2022}, the $\mathcal{CH}$ was used in the classification of Boolean functions which can be computed using non-adaptive measurement-based quantum computation. There, the local operations available to act on the cluster state were contained in a finite level of the Clifford hierarchy.

With all these applications, it is no surprise that the Clifford Hierarchy has been studied abstractly in its own right. We briefly review some of the results here.

The $n$-qubit Clifford Hierarchy \cite{Gottesman1999, Zeng2008} is recursively defined as 

\begin{equation}\label{CHdef}
    \mathcal{CH}_{k} = \{U | U P U^\dagger \subseteq \mathcal{CH}_{k-1}, \forall P\in\mathcal{P}_n\}
\end{equation}
with the first level ($k=1$) defined as $\mathcal{CH}_1 \equiv \mathcal{P}_n$, the $n$-qubit Pauli group. $\mathcal{CH}_2$ is the $n$-qubit Clifford group. 

Throughout this paper we will refer to elements of the $n$-qubit Pauli group as Pauli strings and to elements consisting of tensor products of $X$ ($Z$) and Identity as Pauli $X$ ($Z$) strings. We will often refer to the unitary matrices defined up to a global $U(1)$ phase as gates. We will say that a gate $U$ which satisfies Eqn.~\ref{CHdef} is in the Clifford Hierarchy at the $k$th level. We will also say that a gate is in the Clifford Hierarchy (written $\mathcal{CH}$) if it is in $\mathcal{CH}_{k}$ at any level $k$. It is known that the Clifford Hierarchy is a finite set for any $k$. For $k=1$ we have the Pauli group, for $k=2$ the Clifford group, and for $k\ge 3$ the elements no longer form a group. Note the following relation holds: $\mathcal{CH}_1 \subset \mathcal{CH}_2 \subset ... \mathcal{CH}_k$ since each level contains all elements in the lower levels and contains distinct elements not present in the lower levels \cite{Zeng2008}. 

We will briefly consider how the elements in the set $\mathcal{CH}_k$ for $k\ge 3$ fail to be a group. Associativity is clear since all elements are faithfully represented by unitary operators. The identity operator is in $\mathcal{CH}_1$ and, therefore, all higher levels have identity. The product of elements in $\mathcal{CH}_{k\ge 3}$, however, is not typically in the same level or even in $\mathcal{CH}$ at any level! For example, $HT\in \mathcal{CH}_3$, but $(HT)^2 \notin \mathcal{CH}$. We will discuss some properties of inverses for the cases of semi-Clifford and generalized semi-Clifford later.

As mentioned, the Pauli and Clifford groups are in $\mathcal{CH}$. In Cui \etal~ \cite{Cui2017DiagonalHierarchy} all the diagonal gates in the qubit (and qudit) Clifford Hierarchy were classified. These gates were also shown to form groups at each level in $\mathcal{CH}$. 

Work to elucidate the structure of elements in $\mathcal{CH}$ was undertaken in \cite{Zeng2008, Beigi2010C3Operations,Bengtsson2014OrderHierarchy}. We will briefly introduce semi-Clifford and generalized semi-Clifford gates to better explain the structure of $\mathcal{CH}$ that has been discovered so far.

\begin{definition}[Semi-Clifford Gate]
A gate in $U(2^n)$ is semi-Clifford iff it can be written as $C_1 \prod_j\exp \left( i\theta_j Z_j \right) C_2$ where $C_1,C_2$ are $n$-qubit Clifford gates, $\theta \in [0,2\pi)$, and $Z_j$ are non-identity $Z$ Pauli strings.
\end{definition}

\begin{definition}[Generalized Semi-Clifford Gate]
A gate in $U(2^n)$ is generalized semi-Clifford iff it can be expressed as
\begin{equation}
 U = C_1 P \prod_j\exp \left( i\alpha_j Z_j \right) C_2   
\end{equation}
where $C_1,C_2$ are $n$-qubit Clifford gates, $\theta \in [0,2\pi)$, $Z_j$ are $Z$ Pauli strings, and $P$ is a permutation matrix in $U(2^n)$.
\end{definition}

The following table summarizes what is currently known about the structure of the qubit Clifford Hierarchy:

\begin{table}[h]
\centering
$
\begin{array}{l|c|c|c|c|c}
& k=1  & k=2  & k=3 & k=4 & k\ge5\\
\hline
n=1 & \mathcal{P} & \mathcal{C} & \mathcal{SC} & \mathcal{SC} & \mathcal{SC} \\
\hline
n=2 & \mathcal{P} & \mathcal{C} & \mathcal{SC}  & \mathcal{SC} & \mathcal{SC} \\
\hline
n=3 & \mathcal{P} & \mathcal{C} & \mathcal{SC}  & \neg \mathcal{SC} (\mathcal{GSC}?)  &  \neg \mathcal{SC} (\mathcal{GSC}?)\\
\hline
n=4 & \mathcal{P} & \mathcal{C} & \mathcal{GSC}  & \neg \mathcal{SC} (\mathcal{GSC}?)  & \neg \mathcal{SC} (\mathcal{GSC}?) \\
\hline
n\ge5 & \mathcal{P} & \mathcal{C} & \mathcal{GSC}  & \neg \mathcal{SC} (\mathcal{GSC}?)  & \neg \mathcal{SC} (\mathcal{GSC}?) \\
\end{array}
$
\caption{Here $k$ indicates the level in the Clifford Hierarchy, $n$ the number of qubits, $\mathcal{P}$ indicates the Pauli group, $\mathcal{C}$ indicates the Clifford group, $\mathcal{SC}$ indicates that the elements are semi-Clifford, $\mathcal{GSC}$ indicates that the elements are generalized semi-Clifford, and $\neg \mathcal{SC} (\mathcal{GSC}?)$ indicates that gates are known to exist which are not semi-Clifford (though they are generalized semi-Clifford), but it has not been proven that all such elements must be generalized semi-Clifford. The results presented in this table were proven in \cite{Zeng2008, Beigi2010C3Operations}.}\label{table:CH} 
\end{table}

Much has been accomplished in these studies, but a major conjecture remains open.

\begin{conjecture}[The Generalized semi-Clifford Conjecture]
All elements in the qubit Clifford Hierarchy are generalized semi-Clifford.
\end{conjecture}

Regrettably, we make little progress towards settling this conjecture one way or the other.   

As a preview, Fig.~\ref{fig:CHGroups} shows how some of the groups studied in this paper fit into the Clifford Hierarchy.
\begin{figure}[!h]
     \centering
      \includegraphics[scale=0.8]{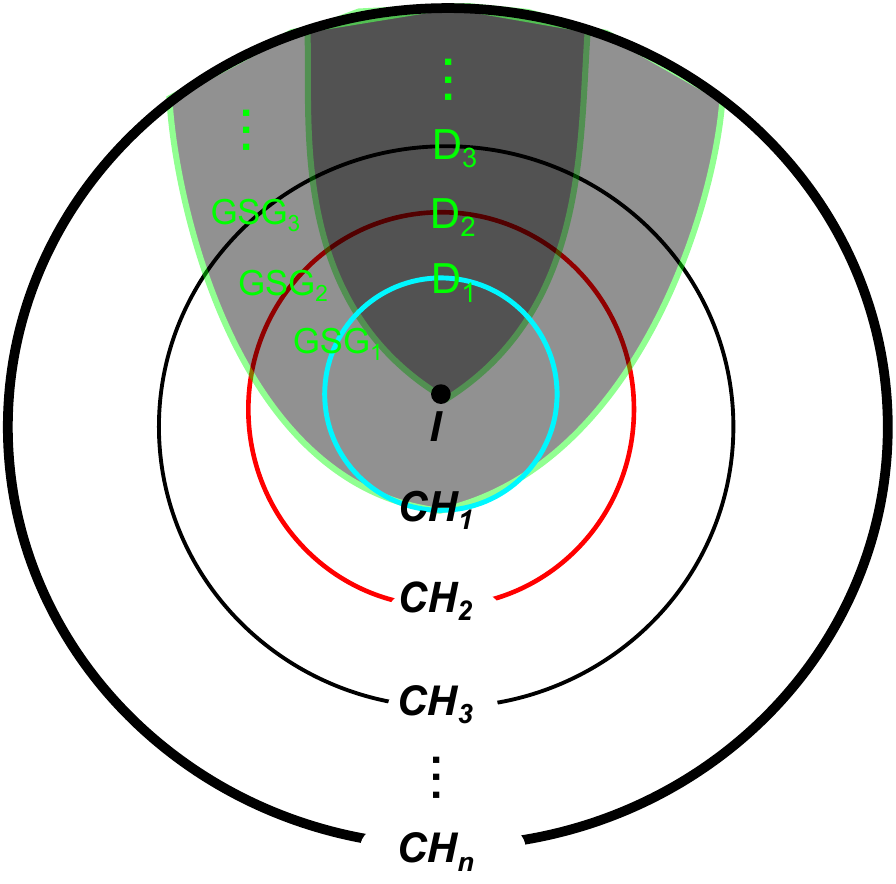}
     \caption{Groups in the Clifford Hierarchy. The blue circle indicates the group $\mathcal{CH}_1$ (the Pauli group), the red circle indicates the group $\mathcal{CH}_2$ (the Clifford group), the diagonal gate groups for each level, $k$, in $\mathcal{CH}$ (classified in \cite{Cui2017DiagonalHierarchy}) are indicated by $\mathcal{D}_k$, and the groups introduced in this paper, $C_{\Pi} \ltimes \mathcal{D}_k$ (a generalized symmetric group), are indicated by $GSG_k$. The set of all gates in the Clifford Hierarchy at level $k$, $\mathcal{CH}_{k\ge 3}$, does not form a group. At every level $\mathcal{D}_k \subset GSG_k$. Note that $GSG_1$ is the Pauli group and, therefore the full Pauli group is a subgroup of $GSG_k$.}
      \label{fig:CHGroups}
\end{figure}

We should also note that the structure of the qudit Clifford Hierarchy has recently been studied \cite{deSilva2021EfficientDimensions}. There, some evidence for similar structure as in the qubit Hierarchy was presented but less has been proven.

\section{Fixing Notation}
Let $\mathcal{P}_1 = \langle i, X, Z\rangle$ be the single-qubit Pauli group. The $n$-qubit Pauli group is then given by $\mathcal{P}_n = \mathcal{P}_1^{\otimes n}$. We refer to elements of the $n$-qubit Pauli group as Pauli strings and to elements in the subgroup of $\mathcal{P}_n$ consisting of tensor products of $X$s ($Z$s) and Identity matrices as Pauli $X$ ($Z$) strings.

The normalizer of the Pauli group, $\mathcal{P}_n$, over $SU(2^n)$ is the $n$-qubit Clifford group. Any element, $C$, in the $n$-qubit Clifford group maps Pauli strings $p\in \mathcal{P}_n$ to $p'\in \mathcal{P}_n$ under conjugation \ie~$CpC^\dagger = p'$. We say that a gate, $g$, (or group of gates, $G$) is `Clifford isomorphic' to another gate $g'$ (a group of gates, $G'$) if $CgC^\dagger = g'$ ($CGC^\dagger = G'$) for some Clifford gate $C$. Similarly we refer to the map, $C( \cdot)C^\dagger$, as a Clifford isomorphism.  
A stabilizer group, $\mathcal{S}$, is an abelian subgroup of $\mathcal{P}_n$ with $-\mathbb{I}\notin \mathcal{S}$. This group is generated by $l\le n$ independent generators (Pauli strings). The number of elements in a stabilizer group is $2^l$, and we will refer to the rank of a stabilizer group as $l$ in this case. When $l=n$, we say that the stabilizer group is full rank.

The diagonal gates in the Clifford Hierarchy, $\mathcal{D}$, were fully classified in Cui \etal~ \cite{Cui2017DiagonalHierarchy}. At any finite level $k$ in $\mathcal{CH}$, they form a finite abelian group denoted $\mathcal{D}_k$. This group is generated by $\frac{\pi}{2^k}$ rotations about all Pauli $Z$ axes. That is, rotations defined by $Z_j[\frac{\pi}{2^k}]\triangleq \exp (i\frac{\pi}{2^k}Z_j)$ where $Z_j$ is any non-identity Pauli $Z$ string. This group is also generated by well-known gates. For $n\ge k$ qubits, the following gates generate $\mathcal{D}_k$\cite{Zeng2008}: 
\begin{equation*}
\langle Z_i[\frac{\pi}{2^k}], \Lambda^1_{i_1,i_2}(Z[\frac{\pi}{2^{k-1}}]),...,\Lambda^{k-1}_{i_1,...,i_k}(Z[\frac{\pi}{2}])\rangle.
\end{equation*}
Here, $\Lambda^k(U)$ denotes the $k$-controlled $U$ gate, a gate with support on $k+1$ qubits; the index $i$ runs over all qubits; $i_1, i_2$ is over all pairs of qubits, \etc~ Note that the order of indices does not matter here since the gates are symmetric. For $n<k$ a similar set of generating gates can be written, but the set of generators must be truncated at $m=n-1$ controls since an $(m-1)$-controlled gate for $m>n-1$ is not supported on $n$ qubits. 

The group of $2^n \times 2^n$ diagonal unitary matrices with $2^l$ root of unity entries denoted $\mbox{Diag}^n_l$ is also useful in this work. We are actually interested in the group $\mbox{Diag}^n_l/U(1)$ since a global phase is unphysical. We `fix' this gauge by requiring the top-left entry of all matrices in the group to be 1. In a slight abuse of notation, we will denote this group by $\mbox{Diag}^n_l$ and only use the latter in the sequel. $\mbox{Diag}^n_l$ is generated by  $\langle Z_i[\frac{\pi}{2^k}], \Lambda^1_{i_1,i_2}(Z[\frac{\pi}{2^k}]),...,\Lambda^{n-1}_{i_1,...,i_n}(Z[\frac{\pi}{2^k}])\rangle$. For large enough $k$ and $l$ any diagonal gate in $\mathcal{CH}$ is an element of $\mathcal{D}_k$ and $\mbox{Diag}_l$, respectively (See Appendix \ref{sec:diag} for more details).

In this work, rotations about Pauli axes play a key role and spend some time introducing our notation below. A rotation about the Pauli axis $\sigma$ by an angle $\theta$ is given by
\begin{equation*}
    R_\sigma(\theta) \triangleq e^{i\sigma\theta}.
\end{equation*}

We denote conjugation by $U$ as
\begin{equation*}
    U R_\sigma(\theta) U^\dagger = R_{U\sigma U^\dagger}(\theta),
\end{equation*}
though the new axis $U\sigma U^\dagger$ is not generally Pauli. However, if $U=C$, a  Clifford unitary, then
\begin{equation*}
    C R_\sigma(\theta) C^\dagger = R_{C\sigma C^\dagger}(\theta) = R_{\sigma'}(\theta),
\end{equation*}
where $\sigma'$ is some Pauli rotation axis.

In this work, we will often encounter products of commuting Pauli rotations (each of which are in the Clifford Hierarchy). We denote a product of such rotations as
\begin{equation*}
    \prod_j R_{\sigma_j}\left(\frac{\alpha_j\pi}{2^{k_j}}\right)
\end{equation*}
where $\alpha_j \in \{0,\pm 1, ..., \pm 2^{k_j}\}$ and all $\sigma_j$ commute (by definition). Since $R_\sigma(\theta) = R_{-\sigma}(-\theta)$ we can always choose the $\sigma_j$ such that they are in some stabilizer group $S$.

We say a Clifford gate, $C$, preserves a stabilizer group, $S$, (and the underlying states) if
\begin{equation*}
    C R_{S_i}(\theta) C^\dagger = R_{S_i'}(\pm\theta),
\end{equation*}
for all $S_i \in S$. Here, $S_i'$ also denotes an element of $S$, and $\pm \theta$ indicates that $C$ may change the sign of $\theta$ and still preserve the stabilizer group.

For example, the stabilizer group of all Pauli $Z$ strings, $S_Z$, is preserved by any Clifford gate generated by $C_{\Pi} = \langle CNOT, X\rangle$ (we refer to these gates as Clifford permutations) and also by any diagonal Clifford gate $C_D = \langle CZ, S \rangle$. The group of all Clifford gates which preserve $S_Z$ is then given by $G_{S_Z}\triangleq C_{\Pi} \ltimes C_D$ where $\ltimes$ indicates a semi-direct product of groups. Elements in $C_D$ commute with all Pauli $Z$ strings (and therefore leave them invariant) and elements in $C_{\Pi}$ can permute Pauli $Z$ strings with other Pauli $Z$ strings in the stabilizer group. These are the only Clifford elements which (via conjugation) preserve the full-rank stabilizer group $S_Z$. 

For other full-rank stabilizer groups (stabilizer groups of maximum size ($|S|=2^n$) on $n$ qubits), $S$, we have that $CS_Z C^\dagger = S$. A Clifford circuit, $C$, always exists which maps all elements of $S_Z$ to another full-rank stabilizer group $S$ (see Lemma \ref{lem:encoding}) and the group of Clifford gates which preserve $S$ is given by $G_{S}=CG_{S_Z}C^\dagger$.  

If $C$ preserves $S$ (\ie~ $C\in G_S$), we have that
\begin{equation*}
    \prod_j C R_{S_j}\left(\frac{\alpha_j\pi}{2^{k_j}}\right)C^\dagger = \prod_j R_{S_j'}\left(\frac{\pm \alpha_j\pi}{2^{k_j}}\right)
\end{equation*}
and $S_j, S_j'\in S$ for all $j$. We can see that all $R_{S_j},R_{S_j'}$ must commute.

In addition, we also look at how (classical) permutations act on Pauli-$Z$ rotations $R_{\sigma_Z}(\theta)$. In what follows, $\Pi^n$ denotes the group of permutations on $2^n$ computational basis states defined on $2^n \times 2^n$ permutation matrices. Note these matrices (gates) are universal for reversible classical computation. A generating set for $\Pi^n$ on $n$ qubits is given by $\langle C^{n-1}(X), SWAP_{i,j}, X_i \rangle$ where $C^{n-1}(X)$ is the $(n-1)$-controlled NOT gate, $SWAP_{i,j}$ is the $SWAP$ gate between any two qubits $i$ and $j$ (though nearest-neighbor $SWAP$s would suffice), and $X_i$ is Pauli $X$ on any qubit $i$ (though a single $X_1$ would suffice). For $Z$ rotations in $\mathcal{CH}$, elements $P \in\Pi$ have the following effect:

\begin{equation*}
    P R_{S_{Z}}\left(\frac{\alpha\pi}{2^{k}}\right) P^\dagger = \prod_j R_{S_{Z_j}}\left(\frac{\alpha_j'\pi}{2^{k_j}}\right).
\end{equation*}
That is, permutations, under conjugation, take diagonal rotations in $\mathcal{CH}$ to products of diagonal rotations in $\mathcal{CH}$. Note that these products are also diagonal matrices. On a fixed number of qubits, $n$, a permutation can only take a diagonal gate at level $k$ in $\mathcal{CH}$ to a product of diagonal gates at level, at most, $k+n$ (see Appendix \ref{sec:diag}). Also, the number of terms $j$ in the product is bounded simply because there are only $2^n-1$ (non-trivial) commuting Pauli $Z$ strings. We note that this mapping is constrained in additional ways such as preserving the spectrum of $R_{S_Z}$, but we do not use these properties here. For clarity, in what follows we denote the classical permutations mentioned above as $\Pi_{S_Z}$ to distinguish them from other groups of permutations.   

Under conjugation by Clifford gates, elements of $C \Pi_{S_Z} C^\dagger = \Pi_{S}$ will act as `permutations' on a new state basis and act on Pauli rotations in $S=CS_Z C^\dagger$ as 
\begin{equation*}
    P R_S\left(\frac{\alpha\pi}{2^{k}}\right) P^\dagger = \prod_j R_{S_j}\left(\frac{\alpha_j'\pi}{2^{k_j}}\right)
\end{equation*}
where $P\in \Pi_{S}$.

Finally, we can take elements from the group of Clifford gates which preserve $S$, $C_S$, and elements from the group of permutation gates which preserve $S$, $\Pi_S$, and examine their combined effect when applied to $R_{S_i}$:
\begin{equation*}
    C P \prod_i R_{S_i}\left(\frac{\alpha\pi}{2^{k_i}}\right) P^\dagger C^\dagger = \prod_j R_{S_j}\left(\frac{\alpha_j'\pi}{2^{k_j}}\right)
\end{equation*}
where $C\in C_S$ and $P\in \Pi_S$.

As above, $S_i, S_j\in S$ for all $i,j$ and this implies that all $R_{S_i},R_{S_j}$ must also commute.

A brief warning: when defining rotations in this section and others, we use the convention that $R_\sigma(\theta) \triangleq e^{i\sigma\theta}$; however when writing circuit diagrams, we use the standard convention which differs from the rotations above by a global phase. For example, on a single qubit $R_Z\left(\frac{\pi}{8}\right) = e^{-i\pi/8}T$. Thus far, this seems like only a global phase, but we will sometimes add controls to these gates which can result in a gate not equal up to a global phase. We have made an effort to stick to the same convention in each section.

\section{Canonical Forms for semi-Clifford and Generalized semi-Clifford gates in $\mathcal{CH}$}

In this section, we introduce canonical forms for both semi-Clifford and generalized semi-Clifford gates in $\mathcal{CH}$. A gate is (generalized) semi-Clifford and in $\mathcal{CH}$ iff it can be written in this form.

A gate in $U(2^n)$ is semi-Clifford iff it can be written as $C_1 \prod_j\exp \left( i\theta_j Z_j \right) C_2$ where $C_1,C_2$ are $n$-qubit Clifford gates, $\theta \in [0,2\pi)$, and $Z_j$ are non-identity $Z$ Pauli strings. A semi-Clifford gate is sometimes said to be `diagonalizable' by left and/or right multiplication with Clifford gates. 

To be in $\mathcal{CH}$, a semi-Clifford gate, $U$, must have the following form:
\begin{equation}\label{eq:scform1}
    U = C_1 \prod_j\exp \left( i\frac{\alpha_j \pi}{2^{k_j}}Z_j \right) C_2
\end{equation}
where $\alpha_j$ are integers and $k_j$ are positive integers. We discuss additional assumptions about the form of $U$ at the end of this section. Equ.~\ref{eq:scform1} is necessary and sufficient since a semi-Clifford gate in the Clifford Hierarchy must be `diagonalizable' by Clifford multiplication, and the resulting diagonal gate must be in the Clifford Hierarchy.  All diagonal gates in the Clifford Hierarchy were classified by Cui \etal~ \cite{Cui2017DiagonalHierarchy}, and we use their classification for qubit diagonal gates here\footnote{We provide an alternative proof of this classification in Appendix \ref{shortProof}.}.

We will often use an equivalent expression for Equ.~\ref{eq:scform1}:
\begin{equation}\label{SCSF}
    U = C_1 C_2 C_2^\dagger \prod_j\exp \left( i\frac{\alpha_j \pi}{2^{k_j}}Z_j \right) C_2 = C \prod_j\exp \left( i\frac{\alpha_j' \pi}{2^{k_j}}S_j \right)
\end{equation}
where $S_j$ indicates a non-identity stabilizer element in a stabilizer group $S$ and $C$ is a Clifford gate. We refer to these elements as stabilizers since they are Pauli strings and must all commute. They are Pauli because Clifford gates map Pauli strings to Pauli strings and must commute since any isomorphism (here $C_2^\dagger(\hspace{1ex}) C_2$) must preserve commutation relations of elements. Note the $S_j$ terms do not need to include all elements in a stabilizer group, $S$; however a stabilizer group must exist which contains all $S_j$. Here, $\alpha_j' = \pm \alpha_j$ and we can always choose $\alpha_j'$ such that the stabilizer elements have the correct $\pm 1$ phase to be in the stabilizer group\footnote{If a stabilizer element requires multiplication by $-1$ to be in the stabilizer group, we can `interpret' a rotation about Pauli axis $\sigma$ by an angle $\alpha$, $R_\sigma(\alpha)$, as the equivalent rotation, $R_{-\sigma}(-\alpha)$, which now has the desired $-1$ factor on the Pauli string.}. 

A gate, $U \in U(2^n)$ is generalized semi-Clifford iff it can be expressed as
\begin{equation}
 U = C_1 P \prod_j\exp \left( i\alpha_j Z_j \right) C_2   
\end{equation}
where $C_1,C_2$ are $n$-qubit Clifford gates, $Z_j$ are $Z$ Pauli strings, and $P$ is a permutation matrix in $U(2^n)$. A generalized semi-Clifford gate is said to be `diagonalizable' by Clifford gates and a permutation. Note that all semi-Clifford gates are generalized semi-Clifford gates; simply set the permutation to Identity. 

To be in $\mathcal{CH}$, it is necessary and sufficient for a generalized semi-Clifford gate, $U$, to have the following form:
\begin{equation}
    U = C_1 P D C_2
\end{equation}
where $D=\prod_j\exp \left( i\frac{\alpha_j \pi}{2^{k_j}}Z_j \right)$ is a diagonal gate in $\mathcal{CH}$ (defined exactly the same as in the semi-Clifford case above) and $P$ indicates a permutation in $\mathcal{CH}$. 

Proof of the canonical form for generalized semi-Clifford gates is somewhat lengthy and it is deferred to Appendix \ref{sec:GsCProof}.

There are no other restrictions on $U$ and a generalized semi-Clifford gate is in $\mathcal{CH}$ iff it can be written in this form.

We will often use the equivalent expression:
\begin{equation} \label{GSCSF}
    U = C_1 C_2 C_2^\dagger P C_2 C_2^\dagger D C_2 = C \Tilde{P} \Sigma
\end{equation}
where $\Sigma=\prod_j\exp \left( i\frac{\pm \alpha_j \pi}{2^{k_j}}S_j \right)$ with all $S_j$ belonging to some stabilizer group $S$, $\Tilde{P}$ permutes basis states stabilized by $S$ as $P$ permuted basis states in the computation basis, and $C=C_1C_2$ is a Clifford gate. Note that $U$ as well as $C$, $\Tilde{P}$, and $S$ are all in $\mathcal{CH}$.

\subsection{Additional assumptions about the canonical form of $U$}
Throughout the paper, unless otherwise noted, we make the following assumptions in regard to (generalized) semi-Clifford gates in the Clifford Hierarchy. We will assume that rotation angles such as $\alpha_j/2^{k_j}$ are given as irreducible fractions and that any Clifford gate (multiples of $\pi/4$ rotations about a Pauli axis) in the diagonal part of $U$ ($\prod_j \exp i\frac{\alpha_j \pi}{2^{k_j}}Z_j$) has been moved to the Clifford part of $U$ (labelled by $C_1$, $C_2$, or $C$ in the equations above) unless otherwise noted. We will also assume that all non-Clifford rotations about the same Pauli axis are combined into a single rotation. After combining terms, the index $j$ in Equ.~\ref{eq:scform1} (or the number of commuting non-Clifford rotations in $S$ in Equ.~\ref{GSCSF}) can then go over at most $2^n -1$ terms (each corresponding to a commuting nontrivial Pauli string). Finally, we assume that no product of commuting non-Clifford rotations is proportional to the identity. If this occurs, the product of rotations can simply be replaced with the identity (possibly times a phase) resulting in an equivalent gate. These assumptions can all be made without loss of generality.  

In the next two sections we will look at the groups that can be formed with semi-Clifford and generalized semi-Clifford gates, respectively. 

\section{Groups of semi-Clifford gates in $\mathcal{CH}$}

While it is well known that gates exist in $\mathcal{CH}$ which are not semi-Clifford, all gates on $n\le 2$ qubits, as well as many other common gates, are semi-Clifford and they warrant some study before moving on to the more general case. Moreover, the structure of these groups is simple to state and comes with fewer caveats than the general case.

In what follows, we will classify the semi-Clifford matrix groups in $\mathcal{CH}$. We make use of many ideas of Englbrecht and Kraus \cite{Englbrecht2020} in this section and the next. 

First, let $U$ be a semi-Clifford gate in $\mathcal{CH}$. Hereafter, we will simply refer to these gates as semi-Clifford, keeping in mind that they are in $\mathcal{CH}$. $U$ can, therefore, be written in canonical form (see Eqn.~\ref{SCSF}). We will examine the additional requirements (beyond Eqn.~\ref{SCSF}) for gates to be generators of a group $G$ of semi-Clifford operators in $\mathcal{CH}$. As we will show, all elements of $G$ will be subject to these constraints. 

Writing $U$ in canonical form we have

\begin{equation*}
    U=C \prod_{j}\exp \left( i\frac{\alpha_{j} \pi}{2^{k_{j}}}S_{j} \right)
\end{equation*}
where $C$ is an $n$-qubit Clifford gate and the terms in the product are commuting non-Clifford gates. If $U\in G$, then $U^2$ must also be in $G$ and we have
\begin{equation*}
    U^2=C \prod_{j_1}\exp \left( i\frac{\alpha_{j_1} \pi}{2^{k_{j_1}}}S_{j_1} \right)C \prod_{j_2}\exp \left( i\frac{\alpha_{j_2} \pi}{2^{k_{j_2}}}S_{j_2} \right) = \Tilde{C} \prod_{j_3}\exp \left( i\frac{\alpha_{j_3} \pi}{2^{k_{j_3}}}S_{j_3} \right),
\end{equation*}
where the RHS is a semi-Clifford gate written in canonical form. This is required since (by our definition of $G$) all elements in $G$ must be semi-Clifford and all semi-Clifford gates can be written in this form.

But, $U^2$ can also be written as

\begin{multline*}
    U^2 = C C C^\dagger \prod_{j_1}\exp \left( i\frac{\alpha_{j_1} \pi}{2^{k_{j_1}}}S_{j_1} \right)C \prod_{j_2}\exp \left( i\frac{\alpha_{j_2} \pi}{2^{k_{j_2}}}S_{j_2} \right) = \\ 
    C^2 \prod_{j_1}\exp \left( i\frac{\alpha_{j_1}' \pi}{2^{k_{j_1}}}S_{j_1}' \right) \prod_{j_2}\exp \left( i\frac{\alpha_{j_2} \pi}{2^{k_{j_2}}}S_{j_2} \right)
\end{multline*}

which must also be semi-Clifford. And this requires that 
\begin{equation*}
     \Tilde{C} \prod_{j_3}\exp \left( i\frac{\alpha_{j_3} \pi}{2^{k_{j_3}}}S_{j_3} \right)= C^2 \prod_{j_1}\exp \left( i\frac{\alpha_{j_1}' \pi}{2^{k_{j_1}}}S_{j_1}' \right) \prod_{j_2}\exp \left( i\frac{\alpha_{j_2} \pi}{2^{k_{j_2}}}S_{j_2} \right).
\end{equation*}
Since all $\alpha_{j},\alpha_{j}'\ne 0$ and $k_j \ge 3$, we can only satisfy this equation when all $S_{j}$ and $S_{j}'$ commute. And by examining how $C$ acts under conjugation on the non-Clifford gates, we see that all $S_{j}$ and $S_{j}'$ commute iff $C^\dagger S_j C \in S$ for some stabilizer group $S$ which contains all $S_{j}$ and $S_{j}'$. In other words, we require that a stabilizer group exists which contains all $S_{j}$ and $S_{j}'$ and this is the case when $C$ is in the normalizer of $S$. The normalizer of $S$ in the Clifford group $C$ is
\begin{equation*}
    N_{C}(S) = \{ g \in C\hspace{1ex}|\hspace{1ex} gSg^{-1} = S\}.
\end{equation*}

We will revisit these constraints on $C$, but first we will show that all elements (not just powers of $U$) in $G$ must have commuting non-Clifford gates. 

Let $U,V$ be semi-Clifford gates (in $\mathcal{CH}$). They can, therefore, be written in canonical form (see Eqn.~\ref{SCSF}). 

Given semi-Clifford gates $U$ and $V$ in canonical form we have

\begin{equation*}
    U=C_1 \prod_{j_1}\exp \left( i\frac{\alpha_{j_1} \pi}{2^{k_{j_1}}}S_{j_1} \right)
\end{equation*}
and 
\begin{equation*}
    V=C_2 \prod_{j_2}\exp \left( i\frac{\alpha_{j_2} \pi}{2^{k_{j_2}}}S_{j_2} \right).
\end{equation*}

To form a group of semi-Clifford gates we must have that $UV\in G$. We can express this requirement as

\begin{equation*}
    UV=C_1 \prod_{j_1}\exp \left( i\frac{\alpha_{j_1} \pi}{2^{k_{j_1}}}S_{j_1} \right)C_2 \prod_{j_2}\exp \left( i\frac{\alpha_{j_2} \pi}{2^{k_{j_2}}}S_{j_2} \right) = C_3 \prod_{j_3}\exp \left( i\frac{\alpha_{j_3} \pi}{2^{k_{j_3}}}S_{j_3} \right),
\end{equation*}
where $C_3 \prod_{j_3}\exp \left( i\frac{\alpha_{j_3} \pi}{2^{k_{j_3}}}S_{j_3} \right)$ is some semi-Clifford gate written in canonical form. 

But, 

\begin{multline*}
    C_1 \prod_{j_1}\exp \left( i\frac{\alpha_{j_1} \pi}{2^{k_{j_1}}}S_{j_1} \right)C_2 \prod_{j_2}\exp \left( i\frac{\alpha_{j_2} \pi}{2^{k_{j_2}}}S_{j_2} \right) \\
    = C_1 C_2 C_2^\dagger \prod_{j_1} \left( i\frac{\alpha_{j_1} \pi}{2^{k_{j_1}}}S_{j_1} \right) C_2 \prod_{j_2}\exp \left( i\frac{\alpha_{j_2} \pi}{2^{k_{j_2}}}S_{j_2} \right) \\
    = C_1 C_2 \prod_{j_1} \left( i\frac{\pm\alpha_{j_1} \pi}{2^{k_{j_1}}}S_{j_1}' \right) \prod_{j_2}\exp \left( i\frac{\alpha_{j_2} \pi}{2^{k_{j_2}}}S_{j_2} \right). 
\end{multline*}

Matching terms, we again see that all non-Clifford rotations about $S_{j_1}'$ must commute with all non-Clifford rotations about $S_{j_2}$. This is equivalent to requiring that all $S_{j_1}'$ and $S_{j_2}$ be part of some stabilizer group $S$. Also, since $U^2$ must also be in $G$, we see that $S_{j_1}$ and $S_{j_1}'$ must commute and, therefore, have Pauli rotations in the same stabilizer group $S$. Similar arguments follow for other products of elements, and we see that all non-Clifford rotations in all elements of $G$ must commute and therefore, be part of some stabilizer group $S$.   

As mentioned earlier, to preserve the stabilizer group $S$, the Clifford part of each term ($C$ in the equations above) must be in the normalizer of $S$. To see why this is necessary, assume that the (potential) generators of a semi-Clifford group $G$ have non-diagonal rotations whose Pauli axes generate a full-rank stabilizer group $S$. Then, assume there exists a generator which in canonical form contains a Clifford gate, $C$, which is not in the normalizer of $S$. Therefore, $C$ must take some non-Clifford rotation by conjugation to a non-Clifford rotation about a Pauli axis not in $S$ and the product of the generators will not be semi-Clifford. We have arrived at a contradiction and must conclude that all $C$ are in the normalizer of $S$. Requiring the Clifford part, $C$, of each term to be in the normalizer of $S$ means that $C$ can only do the following: (1) commute with all elements in $S$ or (2) permute elements of $S$ (while commuting with the rest). Note that we ignore the $\pm 1$ that Clifford gates can apply since it can be absorbed into the rotation angle $\alpha$. 

For each element (expressible as a product of generators) in $G$, we now have shown that each element's commuting non-Clifford rotations must commute with the non-Clifford rotations of all other elements. We proved that this was equivalent to requiring that each Pauli rotation axis for all non-Clifford gates belongs to a common stabilizer group, $S$. Finally, we showed that the Clifford part of each element must be in the normalizer of that group.

In \cite{Cui2017DiagonalHierarchy}, all the qubit diagonal gates contained in $\mathcal{CH}_k$ were shown to form groups at each level, $k$. We denote these groups as $\mathcal{D}^n_k$. Since all the non-Clifford rotations in $G$ commute and are rotations about Pauli axes, we can conjugate each element in $G$ by the same $n$-qubit Clifford gate to diagonalize the non-Clifford rotations. Since these gates are now diagonal and in $\mathcal{CH}$, we can find a group (or subgroup of) $\mathcal{D}^n_k$ (for large enough $k$) such that all conjugated non-Clifford gates are elements of this group. We then refer to the non-Clifford gates in $G$ as `Clifford isomorphic' to a subset of $\mathcal{D}^n_k$. This is generally a subset since the group, as we shall soon see, is formed by Clifford permutations and diagonal gates which do not require that the diagonal gates form a group on their own. Under conjugation, the Clifford gates will now normalize a diagonal (Pauli $Z$ string elements) stabilizer group, but it will also normalize (under conjugation) the group $\mathcal{D}^n_k$. We sum up this discussion in the theorem below.

\begin{theorem}
A set of semi-Clifford gates in $\mathcal{CH}$ generates a group of semi-Clifford gates in $\mathcal{CH}$ iff each element, $l$, in the set, written in canonical form as $C_l \prod_{j_l}\exp \left( i\frac{\alpha_{j_l} \pi}{2^{k_{j_l}}}S_{j_l} \right)$, has all $S_{j_l}$ in a common stabilizer group, $S$, and has all $C_l$ in the normalizer of $S$ in the $n$-qubit Clifford group.
\end{theorem}
This is a restatement of the constraints derived in this section.

The normalizer of the group of diagonal gates in $U(2^n)$ is the group of unitary generalized permutation matrices. These matrices have one non-zero entry (which must be some $e^{i\theta}$ since the matrix is unitary) in each row and column. The normalizer of the diagonal gates in the $n$-qubit Clifford group is the group of diagonal Clifford gates $C_D$ and its normalizer in the Clifford group, the Clifford permutations, $C_{\Pi}$. We can express this group as the semi-direct product $C_{\Pi}\ltimes C_D$. Elements of this group are Clifford gates which map diagonal gates to diagonal gates under conjugation. We say the non-Clifford rotations ($\prod_{j_l}\exp \left( i\frac{\alpha_{j_l} \pi}{2^{k_{j_l}}}S_{j_l}\right)$) in the elements of $G$ have full support on $n$ qubits when the only allowable Clifford gates for any element of $G$ are from a group Clifford-isomorphic to $C_{\Pi}\ltimes C_D$. In this case we also refer to the Clifford gates as restricted. Note that this Clifford group ($C_{\Pi}\ltimes C_D$) is in the normalizer of any diagonal matrix group. Sometimes, however, additional Clifford gates are in the normalizer and we discuss how this can occur below. 

If the non-Clifford gates in all elements in $G$ have no support on some qubit(s) then any Clifford gate can be applied to these qubits in the Clifford portion of any element in $G$. These unconstrained Clifford gates are not always so apparent, and generally a subsystem (of size equal to some number of qubits) can be left unconstrained by the non-Clifford rotations. Additionally, the product of non-Clifford rotations in each element of $G$ is not unique and other products of Pauli rotations in the same stabilizer group can equal the same product. This can further obfuscate the unconstrained subsystem. Checking if any unconstrained subsystems exist is very similar to checking for the existence of a noiseless subsystem \cite{Zanardi1997NoiselessCodes, Lidar1998a, Choi2006MethodSubsystems}. To see this, treat each non-Clifford rotation in each element of $G$ as a Kraus operator in an error channel and then check for the existence of a noiseless subsystem. If an unconstrained subsystem exists, then by Clifford conjugation of each element in $G$ (and possibly expressing some products of Pauli-axis rotations by equivalent products of Pauli-axis rotations), we can write each non-Clifford rotation such that it has no support (acts trivially) on some number of qubits.    

Now, we can use the following lemma to simultaneously diagonalize the non-Clifford rotations in each element of $G$.

\begin{lemma}[Encoding circuit Lemma]\label{lem:encoding}
For every stabilizer group on $n$ qubits, there exists an $n$-qubit Clifford circuit which by conjugation takes all elements of $S$ to Pauli $Z$ strings (a diagonal stabilizer group). Furthermore, if the stabilizer group is on $n$ qubits, but has $2^m$ elements for some $m<n$, then a Clifford circuit exists which by conjugation takes the elements of $S$ to Pauli $Z$ strings on $m$ qubits and to Identity on the remaining $n-m$ qubits. 
\end{lemma}

A Clifford circuit that performs the desired operation is very similar to an encoding circuit for a qubit stabilizer code (or state). In fact, a Clifford circuit exists which (by conjugation) takes any $m$ generators of a $n$-qubit stabilizer group, $S$, to any other $n$-qubit stabilizer group, $S'$, generated by $m$ generators. A Clifford circuit which performs this operation is clearly sufficient to prove our lemma. References to encoding circuits are found in \cite{Nielsen:2000a} and \cite{Gottesman:1997a}. This can also be seen as a consequence of Witt's Theorem as noted in \cite{Heinrich2021Thesis}. Additionally, efficient (in number of qubits $n$) algorithms for constructing such circuits can be found in \cite{Gottesman:1997a}. We will refer to the qubits after conjugation by this Clifford circuit as non-Clifford qubits if non-Clifford rotations have non-trivial support on them and to Clifford qubits otherwise.

Note this Clifford circuit is not unique since we haven't specified how the pure errors (anti-stabilizers) are mapped; however, we do not need uniqueness in any of our constructions, and existence of such a circuit suffices.  

Under conjugation by the Clifford circuit mentioned above, the Clifford part, $C$, of each element in $G$ will map to (1) a diagonal Clifford, (2) to a Clifford permutation, or, if it is unconstrained by the non-Clifford rotations, (3) to some Clifford gate. (1) or (2) are required for any part of $C$ with support on the non-Clifford qubits. Note, that $C$ can map to some combination of (1), (2), and (3) provided it only maps to a combination of (1) and (2) on the non-Clifford qubits.

In some cases Clifford permutations (CNOT gates) can have support on both the Clifford and non-Clifford qubits. This complicates the structure theorem below. We will assume, for now, that conjugation by the Clifford circuit above does not map a Clifford gate to a Clifford permutation with support on both Clifford and non-Clifford qubits. This is discussed further in Appendix \ref{PermsOnC}.

Now, we can state our first structure theorem. 

\begin{theorem}
Every group, $G$, of semi-Clifford gates on $n$ qubits in $\mathcal{CH}$ must be Clifford-isomorphic to a subgroup of \footnote{Up to certain permutations described in Appendix \ref{PermsOnC}.}

\begin{equation}\label{struct1}
    C_{\Pi}^n \ltimes \mathcal{D}_l^n \mbox{ or } (C_{\Pi}^{n-1} \ltimes \mathcal{D}_l^{n-1}) \times C^1  \mbox{ or } ... \mbox{ or } (C_{\Pi}^1 \ltimes \mathcal{D}_l^{1}) \times C^{n-1} \mbox{ or } C^n.
\end{equation}
\end{theorem}

Here we have combined the diagonal Clifford gates into $\mathcal{D}_l$ and denoted the group of Clifford permutations on $n$ qubits by $C_{\Pi}^n$ and the Clifford group on $l$ qubits as $C^l$. The groups $C_{\Pi}^m \ltimes \mathcal{D}_l^{m}$ are sometimes called generalized symmetric groups \cite{Osima1954OnGroup,Puttaswamaiah1969UnitaryGroups}\footnote{Additional properties of these groups are briefly discussed in Appendix \ref{sec:GSG}}. We use the notation $C_{\Pi}^m \ltimes \mathcal{D}_l^{m}$ to indicate the semi-direct product of the $m$-qubit Clifford permutation group (a subgroup of the $2^m \times 2^m$ permutation group) with the diagonal gate group $\mathcal{D}_l^{m}$. It is easily verified that $C_{\Pi}^m$ is in the normalizer of $\mathcal{D}_l^{m}$ and that $C_{\Pi}^m \cap \mathcal{D}_l^{m} = \{I\}$ (the intersection of the groups is trivial).

Examples of semi-Clifford groups are provided in Appendix \ref{sec:examples}.

We conclude this section by giving a recipe that can be used to construct any semi-Clifford group in $\mathcal{CH}$. We will construct groups which have non-Clifford diagonal rotations and note that all other groups are Clifford isomorphic to these. 

\begin{enumerate}
    \item Fix the number of qubits, $0\le m \le n$, which can support the full Clifford group. We will call these the Clifford qubits and refer to the other qubits as non-Clifford qubits.
    \item Choose a set of generating elements $g_i = (\pi_i^{n-m}d_i^{n-m})\otimes c_i^m$ where each $\pi_i^{n-m}$ is a Clifford permutation on the $n-m$ non-Clifford qubits, each $d_i^{n-m}$ is a diagonal gate in $\mathcal{CH}$ on the $n-m$ non-Clifford qubits, and $c_i^m$ is a Clifford gate on the $m$ Clifford qubits. 
    \item If desired add any number of CNOTs with control on the non-Clifford qubits and target on the Clifford qubits to the existing generators or as new generators. See Appendix \ref{PermsOnC} for more details. 
\end{enumerate}

\section{Groups of Generalized semi-Clifford gates in $\mathcal{CH}$}

A generalized semi-Clifford gate is any unitary `diagonalizable' by left/right Clifford multiplication and a (classical) permutation. In this section we will proceed similarly to the last section except that now each element in $G$ is generalized semi-Clifford. We will derive constraints on Clifford and permutations gates set by the the non-Clifford rotations. 

Note conjugation by a non-Clifford permutation does not generally preserve the level of a diagonal gate in $\mathcal{CH}$. In fact, the full group of diagonal gates in $\mathcal{CH}_k$, denoted $\mathcal{D}_k^{m}$, is only preserved under conjugation by Clifford permutations. There are, however, groups of diagonal gates in $\mathcal{CH}$ which have the entire group of permutations (represented by $2^m \times 2^m$ permutation matrices) in their normalizer. These groups have elements represented by $2^m \times 2^m$ diagonal matrices which generate all diagonal matrices with $2^l$ root of unity entries. Since conjugation by a permutation cannot change the eigenvalues of a diagonal matrix and all diagonal matrices which are permutations of these eigenvalues are in this group, we see that the entire permutation group is in the normalizer of this group. It is easy to show that these diagonal matrices are in $\mathcal{CH}$ (see Appendix \ref{sec:diag}) and we denote them by $Diag_l^m$.  

Let $U,V$ be generalized semi-Clifford gates (in $\mathcal{CH}$). They can, therefore, be written in canonical  form (see Eqn.~\ref{GSCSF}). 

With $U$ and $V$ in canonical form we have

\begin{equation*}
    U=C_1 P_1 \prod_{j_1}\exp \left( i\frac{\alpha_{j_1} \pi}{2^{k_{j_1}}}S_{j_1} \right)
\end{equation*}
and 
\begin{equation*}
    V=C_2 P_2 \prod_{j_2}\exp \left( i\frac{\alpha_{j_2} \pi}{2^{k_{j_2}}}S_{j_2} \right)
\end{equation*}
where $P_1$ is a permutation on basis states stabilized by a stabilizer group $S_1$ and $P_2$ is a permutation on basis states stabilized by a stabilier group $S_2$.

Recall that $P_1$ is Clifford isomorphic to a (generally non-Clifford) permutation and it acts on some stabilizer basis $S_1$ (under conjugation) as a classical permutation acts on Pauli $Z$ strings. That is, it can take a single Pauli rotation $\exp \left( i\frac{\alpha_{j_1} \pi}{2^{k_{j_1}}}S_{j_1} \right)$ to a product of Pauli rotations. All Pauli rotations in this product, however, are required to be in $S_1$. All terms in the product of non-Clifford rotations in $U$ and $V$ are required to be in $S_1$ and $S_2$, respectively, since they are written in canonical form. Note that $S_1$ or $S_2$ might not be fully `pinned down' by the permutation; this is just shorthand for any stabilizer group that the permutation preserves.  

To form a group of generalized semi-Clifford gates we must have that all products of $U$ and $V$ are in $G$. We can express this requirement on $UV$ as

\begin{multline*}
    UV=C_1 P_1 \prod_{j_1}\exp \left( i\frac{\alpha_{j_1} \pi}{2^{k_{j_1}}}S_{j_1} \right)C_2 P_2 \prod_{j_2}\exp \left( i\frac{\alpha_{j_2} \pi}{2^{k_{j_2}}}S_{j_2} \right)\\
    = C_3 P_3\prod_{j_3}\exp \left( i\frac{\alpha_{j_3} \pi}{2^{k_{j_3}}}S_{j_3} \right),
\end{multline*}
where $C_3 P_3 \prod_{j_3}\exp \left( i\frac{\alpha_{j_3} \pi}{2^{k_{j_3}}}S_{j_3} \right)$ is some generalized semi-Clifford gate written in canonical form. 

We first look at the constraints on $C_2$ in the product $UV$: 

\begin{multline*}
  C_1 P_1\prod_{j_1}\exp \left( i\frac{\alpha_{j_1} \pi}{2^{k_{j_1}}}S_{j_1} \right)C_2 P_2\prod_{j_2}\exp \left( i\frac{\alpha_{j_2} \pi}{2^{k_{j_2}}}S_{j_2} \right) \\
= C_1 C_2 P_2  \left(P_2^{-1} C_2^\dagger P_1 C_2 P_2\right)\\  \left(P_2^{-1} C_2^\dagger\prod_{j_1} \left( i\frac{\alpha_{j_1} \pi}{2^{k_{j_1}}}S_{j_1} \right) C_2 P_2\right) \prod_{j_2}\exp \left( i\frac{\alpha_{j_2} \pi}{2^{k_{j_2}}}S_{j_2} \right).   
\end{multline*}

For this term to be generalized semi-Clifford, we must have that $C_2 P_2$ is in the normalizer of $S_1$. This requires that $C_2$ must be constrained as in the semi-Clifford case. There must also exist some stabilizer group $S$ which contains all $S_{j_2}$ and all terms $S_{j_1}' = P_2^{-1} C_2^\dagger S_{j_1} C_2 P_2$. Note that a permutation, $P$, by conjugation can generally take a unitary rotation to a product of unitary rotations (or vice versa). This means that the number of stabilizer elements in $S_{j_1}'$ needs not be the same as $S_{j_1}$. A more complicated constraint arises on $P_1' = P_2^{-1} C_2^\dagger P_1 C_2 P_2$. The permutation part of each generator is required (by being written in canonical form) to be a permutation-like gate in $\mathcal{CH}$; however, the set of all permutations in $\mathcal{CH}$ is not a group under matrix multiplication. We must therefore require that all products of permutations $P_2 P_1'$ are also permutations in $\mathcal{CH}$. 

To reiterate, we require that all $S_j$ in all elements in $G$ are in some (shared) stabilizer group $S$, that all permutations $P$ act on the same stabilizer group $S$ (that is are in the normalizer of $S$), and that all permutations also form a group with all elements in $\mathcal{CH}$. The groups of permutations that have this property, to our knowledge, have not been studied in the literature and membership in $\mathcal{CH}$ must be checked for the permutation part of each element in a generalized semi-Clifford group. Note that this is in contrast to the semi-Clifford case where we could enforce the restrictions on the generators and all elements of the group were guaranteed to be semi-Clifford and in $\mathcal{CH}$. We mention a few results in Appendices \ref{sec:examples} \& \ref{sec:GSG} to assist interested parties in constructing these groups.    

\begin{definition}[Maximal Permutation groups in $\mathcal{CH}$]
A permutation group on $n$ qubits, $\tilde{\Pi}^n$, is maximal in $\mathcal{CH}$ if all elements $P \in\tilde{\Pi}^n$ are also in $\mathcal{CH}$ and the addition of any permutation $P' \in \Pi^n$ and $P' \notin \tilde{\Pi}^n$ results in a group which is no longer contained in $\mathcal{CH}$. We refer to these groups as maximal permutation groups for short. 
\end{definition}

Again, we can use the Encoding Circuit Lemma to map (via Clifford conjugation) the commuting non-Clifford, Pauli rotations in our group to commuting Pauli $Z$ (diagonal) rotations. The permutations under this same mapping will become regular (classical) permutation matrices. This new group is Clifford isomorphic to the generalized semi-Clifford groups discussed.

Now, we state our group structure theorem for generalized semi-Clifford gates. 

\begin{theorem}
Every group, $G$, of generalized semi-Clifford gates on $n$ qubits in $\mathcal{CH}$ must be Clifford-isomorphic to a subgroup of \footnote{Up to certain permutations described in Appendix \ref{PermsOnC}.}

\begin{equation}\label{struct2}
    \tilde{\Pi}^n \ltimes \mathcal{D}_l^n \mbox{ or } (\tilde{\Pi}^{n-1} \ltimes \mathcal{D}_l^{n-1}) \times C^1  \mbox{ or } ... \mbox{ or } (\tilde{\Pi}^1 \ltimes \mathcal{D}_l^{1}) \times C^{n-1} \mbox{ or } C^n.
\end{equation}
\end{theorem}

Here we have combined the diagonal Clifford gates into $\mathcal{D}_l$ and denoted any group of permutations in $\mathcal{CH}$ on $n$ qubits by $\tilde{\Pi}^n$ and the Clifford group on $l$ qubits as $C^l$. While the structure of these groups is written similarly to the semi-Clifford group case, there are many distinct groups, $\tilde{\Pi}^n$, which are not subgroups of some larger group of permutations in $\mathcal{CH}$. In other words, each $\tilde{\Pi}^n$ denotes many distinct groups; one for each maximal permutation group on $n$ qubits.    

Examples of generalized semi-Clifford groups are provided in Appendix \ref{sec:examples}. 

We conclude this section by giving a recipe that can be used to construct most generalized semi-Clifford groups in $\mathcal{CH}$. We will construct groups which have non-Clifford diagonal rotations, and note that all other groups are Clifford-isomorphic to these. 

\begin{enumerate}
    \item Fix the number of qubits, $0\le m \le n$, which can support the full Clifford group. We will call these the Clifford qubits and refer to the other qubits as non-Clifford qubits.
    \item Choose a set of generating elements $g_i = (\pi_i^{n-m}d_i^{n-m})\otimes c_i^m$ where each $\pi_i^{n-m}$ is a ($2^{n-m}\times 2^{n-m}$) permutation matrix in $\mathcal{CH}$ which acts on the $n-m$ non-Clifford qubits, each $d_i^{n-m}$ is a diagonal gate in $\mathcal{CH}$ which acts on the $n-m$ non-Clifford qubits, and $c_i^m$ is a Clifford gate on the remaining $m$ Clifford qubits. It is necessary though not sufficient to require that the generators have permutation matrices in $\mathcal{CH}$. In this case we must also require that all elements (products of generators) have permutations in $\mathcal{CH}$. 
    \item Additional $\Lambda^n(X)$ gates ($n$-controlled NOT gates) with control(s) on the non-Clifford qubits and target on the Clifford qubits can be added to the existing generators or as new generators, provided that when combined with the existing permutations they only generate permutations in $\mathcal{CH}$. See Appendix \ref{PermsOnC} for more details. 
\end{enumerate}

This recipe can be used to construct `most' generalized semi-Clifford groups since the Clifford group on the Clifford qubits could be further restricted, allowing for a (potentially) larger set of permutations between (restricted) Clifford and non-Clifford qubits.

\section{Applications: Transversal Gate Classification}

We use a result of O'Conner, Kubica, \& Yoder \cite{OConnor:2018a} as a starting pointing in our examination of possible transversal gates in a qubit QECC. There they show that all transversal gates implementable on a stabilizer code must be in $\mathcal{CH}$ at some finite level. They actually prove a stronger result which allows for low-depth circuits and includes permutations on qubits (relabeling of qubits). Here, we use their result that all such gates must be in $\mathcal{CH}$ and look at the additional constraints that arise from the requirement that all such gates (for a given code) must form a group. 

Without assuming the generalized semi-Clifford Conjecture, we can place restrictions on the permutation gates implementable as transversal gates. In Appendix \ref{P3gates} we classify (up to Clifford multiplication) all non-Clifford permutation gates on 3 qubits and show they are all equivalent to a single Toffoli gate or the Identity. Furthermore, any group of permutations implemented via transversal gates can only have one such Toffoli gate. In other words, for any subset of three logical qubits in an $[[n,k,d]]$ stabilizer code, at most a single distinct Toffoli gate can be implemented transversally.

In Limitations of Transversal Computation through Quantum Homomorphic Encryption, Newman and Shi \cite{Newman2018LimitationsEncryption} show that no stabilizer code can implement a classical reversible universal gate set\footnote{They actually show that (with few exceptions) this is the case for a general quantum error correcting code.}. In that paper, a corollary seems to imply that no such code can implement even a single Toffoli gate transversally. A more careful reading (or simply asking one of the authors) shows \cite{Newman2022PersonalCommunication} that their proof forbids transversal Toffoli gates on any set of three (or more) logical qubits where the target can be applied to any of the logical qubits. They refer to a gate (or set of gates) with this property as a uniform Toffoli gate. 

Stabilizer codes exist which have $\Lambda^2(Z)$ ($CCZ$) transversally and, therefore, a code with logical $X$ and $Z$ definitions `flipped' on one qubit would transversally implement $\Lambda^2(X)$ (Toffoli). Note this code would (likely) no longer implement the $\Lambda^2(Z)$ gate. This shows that, at least for non-Clifford permutation gates, the restrictions we derive are tight. Additional bounds on permutation gates on four or more qubits can likely be proven using the techniques introduced here, though the number of cases to consider increases quickly.

Note that the 15-qubit Reed-Muller code can implement the group $X \ltimes \langle T\rangle$ transversally. Here $X$ is the group $\{I,X\}$ and $\langle T\rangle$ is all products of the $T$ gate. This group is a dihedral group. For a single qubit, the gate $X$ is the only non-trivial permutation gate, and is therefore a maximal permutation group. An interesting question is: do other codes with more than one logical qubit exist which have a transversal gate group with both a maximal permutation group and non-Clifford diagonal gates? Furthermore, if qubit stabilizer codes do not exist which implement such transversal gate groups, what are the additional constraints?

\section{Applications: Efficiently Simulatable Gates}
{\it Note: this section is quite speculative and is included to, hopefully, inspire further research into these topics.}

The maximal groups in the Clifford Hierarchy---groups consisting of elements within the Clifford Hierarchy that cannot be made larger by adding additional elements within the Clifford Hierarchy---are interesting objects warranting future study. Specifically, studying their connection with easy-to-simulate gate sets may prove fruitful. One such maximal subgroup, the $n$-qubit Clifford group, is ubiquitous in quantum information theory and is universal for a complexity class ({\bf parity P} denoted {\bf$\oplus$P}) which is strongly believed to be strictly weaker than the complexity class of polynomial-time circuits, {\bf P} \cite{Aaronson:2004a}. It also has the powerful property that it is a maximal finite subgroup of $SU(2^n)$ meaning that the addition of any unitary gate to the Clifford group generates an infinite group. This raises the question: Do the other maximal groups in the Clifford Hierarchy have interesting properties? Specifically, are they easy to simulate?

In addition to the Clifford group the other class of groups we found in the Clifford Hierarchy were the generalized symmetric groups. Recall these groups have elements which are permutations multiplied by a diagonal gate. The XP stabilizer formalism \cite{Webster2022} can be used to describe single-qubit $X$ and diagonal gates in $\mathcal{CH}$ or even $n$-qubit Pauli X strings and $n$-qubit diagonal gates. General $n$-qubit permutations are efficient the simulate as they are equivalent to the complexity class {\bf P} since these $n$-bit permutations are precisely the universal gates of reversible classical computation. Since many permutations are not in $\mathcal{CH}$, the generalized symmetric groups in $\mathcal{CH}$ are potentially less powerful than {\bf P}. Understanding the complexity (or lack thereof) of simulating these groups is an interesting open question. The lack of a full classification of the permutation gates in $\mathcal{CH}$ is currently a hindrance to this problem.  

 Transversal gates can be thought of as the symmetries of the underlying code \cite{Liu2021}. The Eastin-Knill Theorem \cite{Eastin:2008a} already tells us that no encoded transversal gate set can be universal for quantum computation (the complexity class {\bf BQP}) additionally Newman and Shi \cite{Newman2018LimitationsEncryption} show that a universal reversible gate set cannot be transversal; proving that the transversal gates available to any stabilizer code generate a group of gates that is strictly weaker than {\bf P} could have implications for physical processes thought to be modeled by error correcting codes such as decoding Hawking radiation from black holes \cite{Hayden2007, Yoshida2017, Leone2022}. It should be noted, however, that proving this type of separation is likely to be extremely difficult.  

\section{Summary and Future Work}
We introduced a necessary and sufficient canonical form that semi-Clifford or generalized semi-Clifford elements must satisfy to be in the Clifford Hierarchy. We also identified families of semi-Clifford and later, generalized semi-Clifford gates, which are subgroups of certain generalized symmetric groups. The groups we identified were shown to contain the diagonal gates in the Clifford Hierarchy as a strict subgroup. While we found a complete set of constraints for the generalized semi-Clifford groups, a finer-grained characterization of the permutations in the Clifford Hierarchy would further elucidate their structure.  

Currently all known gates in the Clifford Hierarchy are generalized semi-Clifford, and it is conjectured that this holds for all gates in the Clifford Hierarchy. If true, we have classified (up to possible exceptions mentioned earlier) all groups contained in the Clifford Hierarchy. We note that even if this conjecture is proven false, all groups may still be of this form. To see why this could be true, note that many single elements in $\mathcal{CH}$ such as $HT$ cannot be members of any group in $\mathcal{CH}$ since they generate an infinite group on their own $G=\langle HT \rangle$. If elements exist in $\mathcal{CH}$ which are not generalized semi-Clifford, they may be forced to have such a structure. 

An interesting open problem is to classify groups akin to the generalized symmetric groups found here for qudit stabilizer code. 

\section{Acknowledgements} 

This work stems from a project I started long ago at the Universit\'e de Sherbrooke and was inspired by conversations with David Poulin. I dedicate this work to his memory. 

I'd like to thank Matthew Weippert for a careful reading of an earlier version of this manuscript, Bryan Eastin for providing feedback and for pointing out some subtleties of transversal gates, and Winton Brown for general feedback on the ideas discussed in this manuscript. The errors that remain are entirely my own.

\appendix

\section{Proof of the Generalized semi-Clifford Canonical Form}\label{sec:GsCProof}

The goal of this section is to prove that a permutation gate, $\pi$, and a diagonal gate, $d$, are in the Clifford Hierarchy iff $\pi$ and $d$ are each in the Clifford Hierarchy. Bear with us as we introduce notation and prove some useful lemmas; we will, eventually, prove the stated result.   

We define the group of $n$-qubit diagonal gates with $2^l$ root of unity entries (with $l\ge 1$) as $\mathrm{Diag}^n_l$. This is a finite subgroup of $U(2^n)$ and $\mathrm{Diag}^n_1 \subset \mathrm{Diag}^n_2 \subset \cdots \subset \mathrm{Diag}^n_l$. The permutation group, $\Pi^n$, is the group of $2^n \times 2^n$ permutation matrices (later these same matrices will act on $n$ qubits). $\Pi^n$ is in the normalizer of any group $\mathrm{Diag}^n_l$ and the group $G_l^n = \langle \Pi^n, \mathrm{Diag}^n_l \rangle$ is a semi-direct product of $\mathrm{Diag}^n_l$ and $\Pi^n$ which is typically written as $G_l^n = \Pi^n \ltimes \mathrm{Diag}^n_l$. Note that $\mathrm{Diag}^n_l \cap \Pi^n = \{e\}$ meaning their overlap is trivial. The semi-direct product ensures that any element of $G_l^n$ can be written as $\pi d$ where $\pi\in\Pi^n$ and $d\in \mathrm{Diag}^n_l$. Furthermore, $\pi_1 d_1 = \pi_2 d_2 \iff \pi_1 = \pi_2$ and $d_1 = d_2$. Therefore, any element of $G^n_l$ can be $uniquely$ written as $\pi d$. The inverse of $\pi d = d^\dagger \pi^\dagger = \pi^\dagger (\pi d^\dagger \pi^\dagger)$.

A gate, $U$, is in $\mathcal{CH}_{k+1}$, the Clifford Hierarchy at level $k+1$, if $UPU^\dagger \subseteq \mathcal{CH}_k$ for all Pauli strings $p\in \mathcal{P}$. A gate $U$ is in $\mathcal{CH}$, the Clifford Hierarchy at $any$ level, if $UPU^\dagger \subseteq \mathcal{CH}$ for all Pauli strings $p\in P$. If $UpU^\dagger \nsubseteq \mathcal{CH}$ for any Pauli string $p\in P$, then $U\notin \mathcal{CH}$.

\begin{lemma}
Let $U\in G_{l}^n$, then $UpU^\dagger \in G_{l}^n$ for any $n$-qubit Pauli string, $p$. 
\end{lemma}
We can write $p$ as $p=p_X p_Z$ (up to a phase) and since $p_X$ is a permutation and $p_Z$ is a diagonal gate, we have that $U, U^\dagger$, and $p$ are all in $G_l^n$. $p$ is in $G_1^n$ which is contained in all $G_l^n$ for $l>1$. Since any $G_l^n$ is closed under multiplication, we have that $UpU^\dagger \in G_{l}^n$.  

\begin{lemma}
Let $\pi d \in G_l^n$ and $\pi d \in \mathcal{CH}_1$ (the Pauli group). Then, $\pi = p_X$ for some Pauli $X$ string and $d = p_Z$ for some Pauli $Z$ string.
\end{lemma}
Note that every element in the $n$-qubit Pauli group can be expressed as a combination of a permutation, $p_X$, and a diagonal gate, $p_Z$. The lemma is then a consequence of the semi-direct product.

Since every element in the $n$-qubit Clifford group cannot be expressed as $\pi d$, it is not immediately clear whether $\pi d \in CH_2$ implies that $\pi$ is a Clifford permutation and that $d$ is a diagonal Clifford element. This however can be shown, and we do so in the `bonus lemma' at the end of this section. 

\begin{lemma}\label{lem:sauce}
Let $\pi_k$ be a permutation matrix in $\mathcal{CH}_k$, then $\pi_k p_X \pi_k^\dagger$ is a permutation matrix in $\mathcal{CH}_{k-1}$ for all Pauli $X$ strings, $p_X$. And, $\pi_k p_Z \pi_k^\dagger$ is a gate in $Diag_1$ in $\mathcal{CH}_{k-1}$ for all Pauli $Z$ strings, $p_Z$.
\end{lemma}
$\pi_k \in \mathcal{CH}_k$ implies that $\pi_k p \pi_k^\dagger \in \mathcal{CH}_{k-1}$ for all Pauli strings $p$, but (up to a $\pm 1$ phase) we can write this as $\pi_k p_X p_Z \pi_k^\dagger = \pi_k p_X \pi_k^\dagger \pi_k p_Z \pi_k^\dagger \in \mathcal{CH}_{k-1}$ for all $p_X$ and $p_Z$. Since this must be true even if $p_X = I$ or $p_Z = I$, we conclude that $\pi_k p_X \pi_k^\dagger \in \mathcal{CH}_{k-1}$ for all $p_X$ and that $\pi_k p_Z \pi_k^\dagger \in \mathcal{CH}_{k-1}$ for all $p_Z$. The former is a permutation matrix and the latter is a diagonal gate in $Diag_1$.

From Zeng \etal \cite{Zeng2008} we have the following lemma:
\begin{lemma}\label{lem:CliffEquiv}
$U\in \mathcal{CH}_k$, iff $C_1 U C_2 \in \mathcal{CH}_k$, where $C_1, C_2$ are arbitrary Clifford gates.
\end{lemma}

\begin{lemma}\label{l1}
Any $d\in\mathrm{Diag}^n_l$ is in the Clifford Hierarchy at some level. This is the case, since it can be expressed as $d=\prod_j\exp \left( i\frac{\alpha_j \pi}{2^{k_j}}Z_j \right)$ where $Z_j$ are Pauli $Z$ strings and $\alpha_j$ are integers. When terms with the same Pauli strings are combined and the fractions $\frac{\alpha_j}{2^{k_j}}$ are reduced, the level of $d$ is the maximum $k_j$.
\end{lemma}

In contrast, many elements of $\Pi^n$ for $n\ge 3$ are not in $\mathcal{CH}$ at any level. 

\begin{lemma} \label{lem:diagConj}
Conjugation of an element $d\in \mathrm{Diag}^n_l$ by an element $\pi\in\Pi^n$ yields a diagonal gate in $\mathcal{CH}$ at some (possibly higher) level. $\pi d \pi^{-1}$ must be in $\mathrm{Diag}^n_l$ since $\pi$ is in the normalizer of $\mathrm{Diag}^n_l$. Furthermore, it must be in $\mathcal{CH}$ at some level by Lemma \ref{l1}.  
\end{lemma}

\begin{lemma}\label{lem:pm}
\begin{equation*}
    P_1 e^{i\theta P_2} = 
    \begin{cases}
      e^{i\theta P_2} P_1, & \text{if}\ [P_1, P_2]=0 \\
      e^{-i\theta P_2} P_1, & \text{if}\ \{P_1, P_2\}=0
    \end{cases}
\end{equation*}
where $\theta \in [0,2\pi)$ and $P_1,P_2 \in \mathcal{P}$, the $n$-qubit Pauli group. 

Note that $e^{i\theta P} = I^{\otimes n} \cos{\theta} + iP\sin{\theta}$ and verification of the above equation is straightforward.  
\end{lemma}

\begin{lemma}\label{conjp}
Conjugation of an $n$-qubit Pauli string, $p$, by a diagonal gate (on $n$ qubits) $d_k$ in $\mathcal{CH}_k$ yields $d_kpd_k^{-1} = pd_{k-1} \in \mathcal{CH}_{k-1}$. To see this, note that $d=\prod_j\exp \left( i\frac{\alpha_j \pi}{2^{k_j}}Z_j \right)$ with $k_j \le k$ for all $j$. Also, note that $\prod_j\exp \left( i\frac{\alpha_j \pi}{2^{k_j}}Z_j \right) p = p\prod_j\exp \left( \pm i\frac{\alpha_j \pi}{2^{k_j}}Z_j \right)$ via Lemma \ref{lem:pm} where $+(-)$ indicates that $p$ commutes (anticommutes) with $Z_j$. However a commuting term will be cancelled by its corresponding term in $d^{-1}$ and an anticommuting term will combine with its corresponding term in $d^{-1}$ picking up a factor of 2. Since all $k_j$ are now reduced by at least one, we conclude that $pd_{k-1} \in \mathcal{CH}_{k-1}$. 
\end{lemma}

We define the sets $G^n[k]$ by elements, $\pi d$ in $G^n$ with $\pi \in \mathcal{CH}_k$ and $d$ is any diagonal matrix in $\mathcal{CH}$. Since all permutations in $\mathcal{CH}_k$ are also in $\mathcal{CH}_{k+1}$ we have the following nested relation: $G^n[1]\subseteq G^n[2] \subseteq \cdots \subseteq G^n[k]$.

\begin{theorem}\label{thm:4}
If a permutation gate, $\pi\in\mathcal{CH}$, and a diagonal gate, $d \in \mathcal{CH}$, then  $\pi d \in \mathcal{CH}$.
\end{theorem}

We prove this theorem by induction on the sets $G^n[k]$. Since any permutation in $\mathcal{CH}$ is in $G^n[k]$ for large enough $k$ and all diagonal gates in $\mathcal{CH}$ are in each set, $G^n[k]$, proving the theorem for all $G^n[k]$ proves the stated theorem.

\begin{proof}

Base case $k=2$:
For $k\le 2$ all permutations are Clifford and by Lemma \ref{lem:CliffEquiv}, any Clifford times a (diagonal) gate in $\mathcal{CH}_k$ is in $\mathcal{CH}_k$. Therefore, all $\pi D \in G^n[2]$ are in $\mathcal{CH}$ since all $D \in \mathcal{CH}$ (at some level).

Assume, $\pi d \in G^n[k-1] \implies \pi d \in \mathcal{CH}$.

For elements $\pi_k D \in G^n[k]$ we have that $\pi_k p \pi_k^\dagger \in \mathcal{CH}_{k-1}$ and by Lemma \ref{lem:sauce} we know this implies that $\pi_k p_X \pi_k^\dagger$ is a permutation in $\mathcal{CH}_{k-1}$ for all $p_X$ and $\pi_k p_Z \pi_k^\dagger$ is a diagonal gate in $\mathcal{CH}_{k-1}$ for all $p_Z$.

To prove that $\pi_k D \in \mathcal{CH}$, we must show that $\pi_k D p D^\dagger \pi_k^\dagger$ is in $\mathcal{CH}$ for all $p$. We can write this as: $\pi_k D p D^\dagger \pi_k^\dagger = \pi_k p D' \pi_k^\dagger$ (where $D'$ is a diagonal gate in $\mathcal{CH}$) by Lemma \ref{conjp}. This equals, $\pi_k p_X \pi_k^\dagger (\pi_k p_Z \pi_k^\dagger \pi_k D' \pi_k^\dagger)$ where the product of all gates in the parentheses is a diagonal gate in $\mathcal{CH}$ (for all $p_Z$) by Lemma \ref{lem:diagConj} and by closure of the group $Diag_l^n$ (for some, large enough, $l$). We will denote this product as $D''$ in what follows. By the discussion in the previous paragraph, we know that $\pi_k p_X \pi_k^\dagger$ is a permutation in $\mathcal{CH}_{k-1}$ for all $p_X$. Finally, since $\pi_k p_X \pi_k^\dagger D'' = \pi_{k-1} D'' \in G^n[k-1]$ (for all $p_X$) which by our inductive hypothesis is in $\mathcal{CH}$, we conclude that any element $\pi_k D \in G^n[k]$ is in $\mathcal{CH}$.
\end{proof}

\begin{corollary}\label{Xonly}
When checking if a permutation, $\pi$, is in $\mathcal{CH}$, it suffices only to check if $\pi p_X \pi^{-1} \in \mathcal{CH}$ for all $p_X$. 
\end{corollary}
From the discussion in Thm.~\ref{thm:4} we see that $\pi p_Z \pi^{-1}$ is always a diagonal gate in $\mathcal{CH}$ for any permutation $\pi$ and any Pauli $Z$ string, $p_Z$. Then, $\pi \in \mathcal{CH}$ iff $\pi p_X \pi^\dagger$ is in $\mathcal{CH}$. To see this, note that by Thm.~\ref{thm:4} if $\pi p_X \pi^\dagger \in \mathcal{CH}$, then $\pi p_X \pi^\dagger \pi p_Z \pi^\dagger \in \mathcal{CH}$ for all $p_Z$ and if $\pi p_X \pi^\dagger \notin \mathcal{CH}$, then by definition $\pi \notin \mathcal{CH}$.

This next proof is similar. We define the sets $\tilde{G}^n[k]$ by elements, $\pi d \in \Pi^n \ltimes \mbox{Diag}^n$ with $\pi \in \mathcal{CH}_k$ and $d$ is any diagonal matrix which is $not$ in $\mathcal{CH}$. Since all permutations in $\mathcal{CH}_k$ are also in $\mathcal{CH}_{k+1}$ we have the following nested relation: $\tilde{G}^n[1]\subseteq \tilde{G}^n[2] \subseteq \cdots \subseteq \tilde{G}^n[k]$.

\begin{theorem}\label{thm:5}
If a permutation gate, $\pi\in\mathcal{CH}$, and a diagonal gate, $d \notin \mathcal{CH}$, then  $\pi d \notin \mathcal{CH}$.
\end{theorem}

\begin{proof}
Base case $k=2$:
Here, $\pi$ is Clifford and $d\notin \mathcal{CH}$. Assume that $g=\pi d \in \mathcal{CH}$. Then, $d = \pi^\dagger g$, but $\pi^\dagger$ is a Clifford gate since the Clifford gates form a group, and $\pi^\dagger g \in \mathcal{CH}$ since multiplication by Clifford gate does not change the level in $\mathcal{CH}$. We then have that $d\in \mathcal{CH}$, a contradiction. We conclude that all elements in $\tilde{G}^n[2]$ are not in $\mathcal{CH}$. 

Assume, $\pi d \in \tilde{G}^n[k-1] \implies \pi d \notin \mathcal{CH}$.

Let $\pi d$ be an element in $\tilde{G}^n[k]$. We want to show that there exists some Pauli string $p$ such that $\pi d p d^\dagger \pi^\dagger \notin \mathcal{CH}$. Since $d\notin \mathcal{CH}$, there must exist some $p_X$ (since all $p_Z$ commute with $d$) such that $dp_X d^\dagger \notin \mathcal{CH}$ and this equals $p_X d'$ for some diagonal matrix $d'$ (this is a consequence of Lemma.~\ref{lem:pm}). In what follows, $p_X$ will denote this specific Pauli $X$ string which we will use to show that $\pi d \notin \mathcal{CH}$. We can write $\pi d p_X d^\dagger \pi^\dagger = \pi p_X \pi^\dagger \pi d' \pi^\dagger$. But, $\pi \in \mathcal{CH}_k$ which implies that $\pi p_X \pi^\dagger$ is a permutation in $\mathcal{CH}_{k-1}$. We conclude that $\pi p_X \pi^\dagger \pi d' \pi^\dagger \in \tilde{G}^n[k-1]$ and by our inductive hypothesis, we have that $\pi p_X \pi^\dagger \pi d' \pi^\dagger \notin \mathcal{CH}$ and therefore $\pi d \notin \mathcal{CH}$.
\end{proof}

\begin{corollary}
By Theorems \ref{thm:4} and \ref{thm:5} we have that for $\pi \in \mathcal{CH}$, $\pi d \in \mathcal{CH} \iff d \in \mathcal{CH}$.
\end{corollary}

\begin{theorem}\label{thm:noPi}
If $\pi \notin \mathcal{CH}$ then $\pi d \notin \mathcal{CH}$.
\end{theorem}
\begin{proof}
We prove this theorem by contradiction. 

Let $\pi_1 \notin \mathcal{CH}$. Assume that $\pi_1 d_1 \in \mathcal{CH}_k$ for an arbitrary diagonal gate, $d_1$.  

Since $\pi_1 \notin \mathcal{CH}$ we have that there exists a $p^{1,2}_X$ such that $\pi_1 p^{1,2}_X \pi_1^\dagger = \pi_2 \notin \mathcal{CH}$ where $\pi_2$ is a permutation not necessarily distinct from $\pi_1$. Such a $p^{1,2}_X$ and $\pi_2$ must exist or $\pi_1$ would be in $\mathcal{CH}$. And since $\pi_2 \notin \mathcal{CH}$ we also have that $\pi_2 p^{2,3}_X \pi_2^\dagger = \pi_3 \notin \mathcal{CH}$. And more generally, $\pi_t p^{t,t+1}_X \pi_t^\dagger = \pi_{t+1} \notin \mathcal{CH}$. These permutations are not necessarily distinct and the progression may `loop', but we can continue the progression (even if a loop occurs) to arbitrary length $t$, and each $\pi_t$ will not be in $\mathcal{CH}$.

Now, if $\pi_1 d_1 \in \mathcal{CH}_k$ we have that $\pi_1 d_1 p d_1^\dagger \pi_1^\dagger \in \mathcal{CH}_{k-1}$ for all Pauli strings $p$. But $p^{1,2}_X$ is one such Pauli string and therefore
\begin{equation*}
\pi_1 d_1 p^{1,2}_X d_1^\dagger \pi_1^\dagger = \pi_1 p^{1,2}_X \pi_1^\dagger \pi_1 d_1' \pi_1^\dagger = \pi_2 d_2   
\end{equation*}
where $d_1'$ is given by Lemma \ref{conjp} and $d_2 = \pi_1 d_1' \pi_1^\dagger$. And by assumption $\pi_2 d_2 \in \mathcal{CH}_{k-1}$. We can proceed in this manner choosing $p_X^{t,t+1}$ since it must be true for all Pauli strings, $p$. Whatever level $k$ we started with, we will eventually get to level $k=1$ and there we have that $\pi_t d_t \in \mathcal{CH}_1$. Since this must be a Pauli group element, $P$, which can be expressed (up to a phase) as a product of $p_X p_Z$, where $p_Z$ is a diagonal Pauli gate, and $p_X$ is a permutation (Pauli) gate, we have that $\pi_t d_t = p_X p_Z$ and via the splitting property $\pi_t = p_X$ and $d_t = p_Z$. But $\pi_t \notin \mathcal{CH}$ (at any level) and $p_X \in \mathcal{CH}_1$. We have a contradiction and our assumption that $\pi_1 d_1 \in \mathcal{CH}_k$ must not be true. Therefore, $\pi_1 d_1 \notin \mathcal{CH}_k$. But since this proof works for arbitrary $k$, we conclude that $\pi_1 d_1 \notin \mathcal{CH}$.   
\end{proof}

\begin{corollary}
By Thms \ref{thm:noPi}, \ref{thm:4}, and \ref{thm:5} we conclude that $\pi d \in \mathcal{CH}$ iff $\pi \in \mathcal{CH}$ and $d \in \mathcal{CH}$.
\end{corollary}

This establishes the canonical form for generalized semi-Clifford gates used in this paper. We can add left and right Clifford operators to any $\pi d \in \mathcal{CH}$ since by Lemma \ref{lem:CliffEquiv}, this is also in $\mathcal{CH}$ at the same level as $\pi d$.

\begin{corollary}\label{cor:pi_k}
If $\pi d \in \mathcal{CH}_k$, then $\pi \in \mathcal{CH}_k$.
\end{corollary}
This follows from a similar proof as in Thm.~\ref{thm:noPi}. Let $\pi d \in \mathcal{CH}_k$. Assume that $\pi \notin \mathcal{CH}_k$. It must be in $\mathcal{CH}$ by Thm.~\ref{thm:noPi} (or else we would have already reached a contradiction). Assume $\pi$ is in $\mathcal{CH}$ at some level higher than $k$. We will use $k+1$, but any higher level will work. Then, by a similar  argument as Thm.~\ref{thm:noPi}, we can always choose $p_X^{k+1,k}$ such that $\pi p_X^{k+1,k} \pi^\dagger \in \mathcal{CH}_k$. We eventually get to the case where $\pi p_X^{3,2} \pi^\dagger$ is Clifford (or higher if not using $k+1$), but must be in the Pauli group since $\pi d \in \mathcal{CH}_k$. We arrive at a similar contradiction which implies that $\pi \in \mathcal{CH}_k$.

\begin{lemma}[Bonus Lemma]\label{lem:bonus}
If $\pi d \in \mathcal{CH}_2$, then $\pi \in \mathcal{CH}_2$ and $d \in \mathcal{CH}_2$.
\end{lemma}
Let $\pi d \in \mathcal{CH}_2$, then $\pi d p d^\dagger \pi^\dagger \in \mathcal{CH}_1$ for all Pauli strings $p$. Then, by \ref{cor:pi_k} $\pi \in \mathcal{CH}_2$. Since $\pi d = C$, a Clifford gate, we have that $d = \pi^\dagger C$ which is also Clifford on account of the Clifford elements forming a group. This implies that $d$ is also Clifford which completes our proof.

Note that $\pi \in \mathcal{CH}_k$ and $d \in \mathcal{CH}_k$ do not imply that $\pi d \in \mathcal{CH}_k$ as the following circuit identity shows:

\begin{equation*}
\begin{tabular}{c}
\vspace{.2em}
\Qcircuit @C=0.5em @R=1.3em {
&\qw & \ctrl{1} & \qw & \qw & \qw & \ctrl{1} & \qw \\
&\qw & \ctrl{1} & \qw & \qw & \qw & \ctrl{1} & \qw \\
&\qw & \targ & \gate{T} & \gate{X} & \gate{T^\dagger} & \targ & \qw \\
}
\vspace{.2em}\hspace{1.2em}
\\
\end{tabular}
=
\begin{tabular}{c}
\vspace{.2em}
\Qcircuit @C=0.5em @R=1.3em {
&\qw & \ctrl{1} & \qw & \ctrl{1} & \qw \\
&\qw & \ctrl{1} & \qw & \ctrl{1} & \qw \\
&\qw & \targ & \gate{S} & \targ & \gate{X} \\
}
\vspace{.2em}\hspace{1.2em}
\\
\end{tabular}
=
\begin{tabular}{c}
\vspace{.2em}
\Qcircuit @C=0.5em @R=1.3em {
&\qw & \ctrl{1} & \ctrl{1} & \qw \\
&\qw & \gate{S} & \ctrl{1} & \qw \\
&\qw & \gate{S} & \ctrl{-1} & \gate{X} \\
}
\vspace{.2em}\hspace{1.2em}
\\
\end{tabular}
\notin \mathcal{CH}_2. 
\end{equation*}

Here both $T$ and Toffoli are in $\mathcal{CH}_3$, but the product as illustrated in the circuit above is in $\mathcal{CH}_4$. Also, while we have shown that $\pi d \in \mathcal{CH}_k$ implies $\pi \in \mathcal{CH}_k$, we have not proven that $\pi d \in \mathcal{CH}_k$ implies $d\in \mathcal{CH}_k$, only that it implies $d\in \mathcal{CH}$ (at some level). We leave it as an open problem to prove this or to find a counterexample. A counterexample would consist of a gate $\pi d \in \mathcal{CH}_k$ where $d$ is a diagonal gate in $\mathcal{CH}$, but not in $\mathcal{CH}_{k}$. 

\section{Some Examples of Groups in the Clifford Hierarchy}\label{sec:examples}
For one and two qubits we can use the constraints derived in this paper to find all groups in $\mathcal{CH}$ satisfying the constraints. 

\subsection*{Case: $n=1$}

For a single qubit the only permutation (Clifford or otherwise) is Pauli $X$ and the only semi-Clifford groups in $\mathcal{CH}$ are Clifford-isomorphic to subgroups of 
\begin{equation*}
    \left\langle X, Z\left [\frac{\pi}{2^k} \right]\right\rangle \mbox{ or } C^1.
\end{equation*}
The former is a generalized symmetric group which is also a dihedral group corresponding to the symmetries of a $2^k$-sided polygon. The latter is the one-qubit Clifford group. It has been proven that all gates in $\mathcal{CH}$ on one-qubit are semi-Clifford and therefore all single-qubit groups in $\mathcal{CH}$ must be Clifford-isomorphic to a subgroup of one of these groups. It is also known\cite{OConnor:2018a} that all transversal logical gates in qubit stabilizer codes must be in $\mathcal{CH}$. Since transversal gates must also form a group, the only transversal logical gates for an $[[n,k=1,d\ge 2]]$\footnote{The code should also be Bell-pair free and free of trivially-encoded qubits.} qubit stabilizer code must be Clifford-isomorphic to a subgroup of one of the groups shown above. This is similar to the result proven in \cite{Wirthmuller2011AutomorphismsCodes}. Note that Wirthm\"uller's result is more elementary in that it does not use disjointness results (or the cleaning lemma) to prove their claims.

\subsection*{Case: $n=2$}
    
    For two qubits, all elements in $\mathcal{CH}$ are semi-Clifford and all groups are Clifford-isomorphic to a subgroup of one of the following groups
    \begin{equation*}
    \left\langle C^2_{\Pi}, Z_1\left [\frac{\pi}{2^{k_1}}\right], Z_2\left [\frac{\pi}{2^{k_2}}\right], Z_1Z_2\left [\frac{\pi}{2^{k_{12}}} \right]\right\rangle \mbox{ or } \left\langle CX_{(1,2)},X_1, Z_1\left [\frac{\pi}{2^k}\right], H_2, S_2 \right\rangle \mbox{ or } C^2.
    \end{equation*}
    Here $C^2_{\Pi}=\langle CX_{(1,2)}, CX_{(2,1)}, X_1, X_2\rangle$ is the two-qubit Clifford-permutation group, $H$ and $S$ generate the one-qubit Clifford group, and $C^2$ is the two-qubit Clifford group.

\subsection*{Case: $n\ge 3$}

For $n\ge 3$, we have non-Clifford permutations and the constraints become more complicated. Furthermore, it has not been proven that all elements in $\mathcal{CH}$ are generalized semi-Clifford. Even if we assume the generalized semi-Clifford conjecture, the number of distinct groups grows rapidly with $n$ and finding (and listing) them all would quickly become unsustainable. We can, however, still provide some examples of generalized semi-Clifford groups.  

First, we identify some permutation groups in $\mathcal{CH}$ which have non-Clifford permutations. In the groups presented below, $C_{\Pi}^{[a\mbox{-}b]}$ indicates the full group of Clifford permutations on qubits $a$-$b$ (with 1 being the topmost qubit and proceeding downwards). It can be verified (with some work) that all permutations generated by these groups are in $\mathcal{CH}$.

\begin{center}
\begin{equation*}
G_{\tilde{\Pi}_3} = \left \langle
\begin{tabular}{c}
\Qcircuit @C=0.5em @R=1.3em {
&\ctrlblankA{1} \\
&\ctrlblankA{1} \\
&\targblank \\
}
\end{tabular},
C_{\Pi}^{[1\mbox{-}2]}
\right \rangle
\end{equation*}
\vspace{1ex}\\
\begin{equation*}
G_{\tilde{\Pi}_{4a}} = \left \langle
\begin{tabular}{c}
\Qcircuit @C=0.5em @R=1.3em {
& \ctrlblankA{1} \\
& \ctrlblankA{1} \\
& \ctrlblankA{1} \\
& \targblank \\
}
\end{tabular},
\begin{tabular}{c}
\Qcircuit @C=0.5em @R=1.3em {
& \ctrlblankA{1} \\
& \ctrlblankA{1} \\
& \targblank \\
 \\
}
\end{tabular},
C_{\Pi}^{[1\mbox{-}2]}
\right \rangle
\end{equation*}
\vspace{1ex}\\
\begin{equation*}
G_{\tilde{\Pi}_{4b}} = \left \langle
\begin{tabular}{c}
\Qcircuit @C=0.5em @R=1.3em {
& \ctrlblankA{1} \\
& \ctrlblankA{1} \\
& \ctrlblankA{1} \\
& \targblank \\
}
\end{tabular},
C_{\Pi}^{[1\mbox{-}3]}
\right \rangle
\end{equation*}
\end{center}

We also give some examples of a set of (gate) generators for $Diag_l^n$ the $2^n \times 2^n$ diagonal unitary matrices with $2^l$ root of unity entries:  

\begin{equation*}
    Diag_3^1 = \langle T \rangle,
\end{equation*}

\begin{equation*}
    Diag_3^2 = \langle T_1, T_2, \Lambda^1(T) \rangle,
\end{equation*}

\begin{equation*}
    Diag_3^3 = \langle T_1, T_2, T_3, \Lambda^1_{1,2}(T), \Lambda^1_{1,3}(T), \Lambda^1_{2,3}(T), \Lambda^2(T)  \rangle,
\end{equation*}
and finally
\begin{equation*}
    Diag_3^4 = \langle T_i, \Lambda^1_{i,j}(T), \Lambda^2_{i,j,k}(T), \Lambda^3(T)  \rangle.
\end{equation*}
Here $i$ indicates any single qubit, $i,j$ any pair of qubits, \etc 

We can construct new groups with the semi-direct product; since these groups (and elements) satisfy the constraints to be a group in $\mathcal{CH}$, we have that any subgroup of the following are in $\mathcal{CH}$:  
\begin{equation*}
    G_{\tilde{\Pi}_3} \ltimes Diag_3^3,\hspace{1ex}
    G_{\tilde{\Pi}_{4a}} \ltimes Diag_3^4,\hspace{1ex}\mbox{ or }\hspace{1ex}
    G_{\tilde{\Pi}_{4b}} \ltimes Diag_3^4.
\end{equation*}

These are just examples which can be extended to $Diag_l$ with higher (larger $l$) roots of unity or to different groups of permutations or to permutations on more qubits.  

\section{Permutation Gates in $\mathcal{CH}$ on 3 qubits}\label{P3gates}

Earlier, we showed that a generalized semi-Clifford gate must have a permutation in the Clifford Hierarchy to be in the Clifford Hierarchy itself. To get an idea of how restrictive this requirement on permutations is, we find all the classical permutations on three qubits (that is $8\times 8$ permutation matrices) which are in $\mathcal{CH}$. Note that all Clifford permutations must be in $\mathcal{CH}$ and we will look only at the non-Clifford permutations here. For permutations on three qubits, these are generated by Toffoli gates. We will associate the Toffoli with target on qubit 1 with generator `$a$', the Toffoli with target on qubit 2 with generator `$b$', and the Toffoli with target on qubit 3 with the generator `$c$'. With this association, we have the following group presentation:
\begin{equation*}
    \langle a,b,c \hspace{0.5ex}|\hspace{0.5ex} g_i^2 = 1, (g_i g_j)^3 = 1, (g_i g_j g_k)^4 = 1, g_i g_j g_i g_k = g_k g_i g_j g_i \rangle.
\end{equation*}

Here $g_i,g_j,g_k$ can represent any generator, but $g_i \ne g_j \ne g_k \ne g_i$. The group relations are given by easily verified Toffoli gate identities. This group is finite and all distinct words are given below:

\begin{multline*}
    \{1, a, b, c, ab, ac, ba, bc, ca, cb, aba, abc, aca, acb, bac, bca, bcb, cab, cba,\\ 
    abac, abca, abcb, acab, acba, bacb, bcab, cabc, cbac, abcab, acbac, bacba, bcabc,\\ 
    cabca, cbacb, abcabc, acbacb, bacbac \}.
\end{multline*}

Many of these words are equivalent up to Clifford operations and we have the following classes of non-Clifford (and non-trivial) 3-qubit permutations: 
\begin{equation*}
a, ab, aba, abc, abac, abca, abcab, abcabc.    
\end{equation*}

To show that a gate, $U$, is not in $\mathcal{CH}$ at any level it suffices to show that 
\begin{itemize}
    \item (1) $UpU^\dagger = C_L U C_R$ for at least one Pauli string $p$ \\
    or that
    \item (2) $UpU^\dagger = C_L V C_R$ for at least one Pauli string $p$ where $V$ is a gate not in $\mathcal{CH}$. 
\end{itemize}

We already know that a single Toffoli is in $\mathcal{CH}_3$. Then, using the relation:

\begin{center}
\begin{tabular}{c}
\Qcircuit @C=1.5em @R=1.3em {
& \ctrl{1} \qw & \ctrl{1} \qw & \ctrl{1} \qw & \qw  \\
& \ctrl{1} \qw & \targ & \ctrl{1} \qw & \qw  \\
& \targ & \ctrl{-1} \qw & \targ & \qw  \\
}
\end{tabular}
 = 
\begin{tabular}{c}
\Qcircuit @C=1.5em @R=1.3em {
& \qw & \ctrl{1} \qw & \qw & \qw  \\
& \ctrl{1} \qw & \targ & \ctrl{1} \qw & \qw  \\
& \targ & \ctrl{-1} \qw & \targ & \qw  \\
}
\end{tabular}
=
\begin{tabular}{c}
\Qcircuit @C=1.4em @R=2em {
& \ctrl{1} \qw & \qw \\
& \qswap \qw & \qw  \\
& \qswap \qwx & \qw  \\
}
\end{tabular},
\end{center}

we see that permutations in $aba$ class are also in $\mathcal{CH}_3$ since they are Clifford-equivalent to a single Toffoli. In Zeng \etal \cite{Zeng2008} they showed that

\begin{center}
\begin{tabular}{c}
\Qcircuit @C=1.5em @R=1.3em {
& \ctrl{1} \qw & \ctrl{1} \qw & \gate{X} & \ctrl{1} \qw & \ctrl{1} \qw & \qw \\
& \targ  & \ctrl{1} \qw & \qw & \ctrl{1} \qw & \targ & \qw  \\
& \ctrl{-1} & \targ & \qw & \targ & \ctrl{-1} \qw & \qw \\
}
\end{tabular}
 = 
\begin{tabular}{c}
\Qcircuit @C=1.5em @R=1.3em {
& \gate{X} & \qw & \ctrl{1} \qw & \ctrl{1} \qw & \qw\\
& \targ  & \ctrl{1} \qw & \targ & \ctrl{1} \qw & \qw  \\
& \ctrl{-1} & \targ & \ctrl{-1} \qw & \targ & \qw \\
}
\end{tabular}
\end{center}

and 

\begin{center}
\begin{tabular}{c}
\Qcircuit @C=0.7em @R=1.3em {
& \qw & \targ & \ctrl{1} \qw & \ctrl{1} \qw & \gate{X} & \ctrl{1} \qw & \ctrl{1} \qw & \targ & \qw \\
& \qw & \ctrl{-1} \qw & \targ  & \ctrl{1} \qw & \qw & \ctrl{1} \qw & \targ & \ctrl{-1} \qw & \qw  \\
& \qw & \ctrl{-1} \qw & \ctrl{-1} & \targ & \qw & \targ & \ctrl{-1} \qw & \ctrl{-1} \qw & \qw \\
}  
\end{tabular}
=
\begin{tabular}{c}
\Qcircuit @C=0.7em @R=1.3em {
& \qw & \qw & \targ & \qw & \targ & \ctrl{1} \qw & \ctrl{1} \qw & \targ & \targ & \gate{X} & \qw\\
& \qw & \targ  & \qw & \ctrl{1} \qw & \qw & \targ & \ctrl{1} \qw & \qw & \ctrl{-1} \qw & \qw & \qw  \\
& \qw & \ctrl{-1} \qw & \ctrl{-2} \qw & \targ & \ctrl{-2} \qw & \ctrl{-1} \qw & \targ & \ctrl{-2} \qw & \qw & \qw & \qw \\
}
\end{tabular}
\end{center}

which shows that $ab$ and $abc$ classes are not in $\mathcal{CH}$. 

Note that the example for class $abc$ also proves that class $abca$ is not in $\mathcal{CH}$ since

\begin{center}
\begin{tabular}{c}
\Qcircuit @C=0.7em @R=1.3em {
& \qw & \targ & \ctrl{1} \qw & \ctrl{1} \qw & \targ & \gate{X} & \targ & \ctrl{1} \qw & \ctrl{1} \qw & \targ &  \qw \\
& \qw & \ctrl{-1} \qw & \targ  & \ctrl{1} \qw & \ctrl{-1} \qw & \qw & \ctrl{-1} \qw & \ctrl{1} \qw & \targ & \ctrl{-1} \qw & \qw  \\
& \qw & \ctrl{-1} \qw & \ctrl{-1} & \targ & \ctrl{-1} \qw & \qw & \ctrl{-1} \qw & \targ & \ctrl{-1} \qw & \ctrl{-1} \qw & \qw \\
}
\end{tabular}
 = 
\begin{tabular}{c}
\Qcircuit @C=0.7em @R=1.3em {
& \qw & \targ & \ctrl{1} \qw & \ctrl{1} \qw & \gate{X} & \ctrl{1} \qw & \ctrl{1} \qw & \targ & \qw \\
& \qw & \ctrl{-1} \qw & \targ  & \ctrl{1} \qw & \qw & \ctrl{1} \qw & \targ & \ctrl{-1} \qw & \qw  \\
&\qw & \ctrl{-1} \qw & \ctrl{-1} & \targ & \qw & \targ & \ctrl{-1} \qw & \ctrl{-1} \qw & \qw \\
}
\end{tabular}.
\end{center}

For class $abac$ we have the following circuit identity:

\begin{center}
\begin{tabular}{c}
\Qcircuit @C=0.7em @R=1.3em {
&\qw & \ctrl{1} \qw & \targ & \ctrl{1} \qw & \ctrl{1} \qw & \gate{X} & \ctrl{1} \qw & \ctrl{1} \qw & \targ & \ctrl{1} \qw &  \qw \\
&\qw & \targ & \ctrl{-1} \qw  & \targ & \ctrl{1} \qw & \qw & \ctrl{1} \qw & \targ & \ctrl{-1} \qw & \targ & \qw  \\
&\qw & \ctrl{-1} \qw & \ctrl{-1} & \ctrl{-1} \qw & \targ & \qw & \targ & \ctrl{-1} \qw & \ctrl{-1} \qw & \ctrl{-1} \qw & \qw \\
} 
\end{tabular}
 = 
\begin{tabular}{c}
\Qcircuit @C=0.7em @R=1.3em {
& \qw & \gate{X} & \qw & \qswap & \targ & \ctrl{1} \qw & \ctrl{1} & \targ \qw & \qswap & \targ & \qw \\
& \qw & \ctrl{1} \qw & \targ & \qw \qwx & \ctrl{-1} \qw & \targ & \ctrl{1} \qw & \ctrl{-1} \qw & \qw \qwx & \qw & \qw  \\
& \qw & \targ & \ctrl{-1} & \qswap \qwx &  \ctrl{-1} \qw & \ctrl{-1} \qw & \targ & \ctrl{-1} \qw & \qswap \qwx & \ctrl{-2} \qw & \qw \\
}
\end{tabular}
\end{center}

For class $abcab$ we have the following identity:

\begin{tabular}{c}
\Qcircuit @C=0.5em @R=1.3em {
& \qw & \targ & \ctrl{1} \qw & \ctrl{1} \qw & \targ & \ctrl{1} \qw & \qw & \ctrl{1} \qw & \targ & \ctrl{1} \qw & \ctrl{1} \qw & \targ &  \qw \\
& \qw & \ctrl{-1} \qw & \targ  & \ctrl{1} \qw & \ctrl{-1} \qw & \targ & \qw & \targ & \ctrl{-1} \qw & \ctrl{1} \qw & \targ & \ctrl{-1} \qw & \qw  \\
& \qw & \ctrl{-1} \qw & \ctrl{-1} & \targ & \ctrl{-1} \qw & \ctrl{-1} \qw & \gate{X} & \ctrl{-1} \qw & \ctrl{-1} \qw & \targ & \ctrl{-1} \qw & \ctrl{-1} \qw & \qw \\
} 
\end{tabular}
 = 
\begin{tabular}{c}
\Qcircuit @C=0.5em @R=1.3em {
& \qw & \ctrl{1} \qw & \targ & \qw & \qw & \ctrl{1} \qw & \targ & \ctrl{1} \qw & \ctrl{1} \qw & \qw & \qw & \targ & \qw \\
& \qw & \targ & \ctrl{-1} \qw & \ctrl{1} \qw & \qswap & \targ & \ctrl{-1} \qw & \targ & \ctrl{1} \qw & \qswap & \ctrl{1} \qw & \ctrl{-1} \qw & \qw  \\
& \qw & \gate{X} & \qw & \targ & \qswap \qwx & \ctrl{-1} \qw & \ctrl{-1} \qw & \ctrl{-1} \qw & \targ & \qswap \qwx & \targ & \qw & \qw \\
}
\end{tabular}

\vspace{1em}
And finally, for class $abcabc$ the example from class $abcab$ also shows it is not in the Clifford Hierarchy.

\begin{tabular}{c}
\Qcircuit @C=0.4em @R=1.3em {
& \qw & \targ & \ctrl{1} \qw & \ctrl{1} \qw & \targ & \ctrl{1} \qw & \ctrl{1} \qw & \qw & \ctrl{1} \qw & \ctrl{1} \qw & \targ & \ctrl{1} \qw & \ctrl{1} \qw & \targ &  \qw \\
& \qw & \ctrl{-1} \qw & \targ  & \ctrl{1} \qw & \ctrl{-1} \qw & \targ & \ctrl{1} \qw & \qw & \ctrl{1} \qw & \targ & \ctrl{-1} \qw & \ctrl{1} \qw & \targ & \ctrl{-1} \qw & \qw  \\
& \qw & \ctrl{-1} \qw & \ctrl{-1} & \targ & \ctrl{-1} \qw & \ctrl{-1} \qw & \targ & \gate{X} & \targ & \ctrl{-1} \qw & \ctrl{-1} \qw & \targ & \ctrl{-1} \qw & \ctrl{-1} \qw & \qw \\
} 
\end{tabular}
= 
\begin{tabular}{c}
\Qcircuit @C=0.4em @R=1.3em {
& \qw & \targ & \ctrl{1} \qw & \ctrl{1} \qw & \targ & \ctrl{1} \qw & \qw & \ctrl{1} \qw & \targ & \ctrl{1} \qw & \ctrl{1} \qw & \targ &  \qw \\
&\qw & \ctrl{-1} \qw & \targ  & \ctrl{1} \qw & \ctrl{-1} \qw & \targ & \qw & \targ & \ctrl{-1} \qw & \ctrl{1} \qw & \targ & \ctrl{-1} \qw & \qw  \\
& \qw & \ctrl{-1} \qw & \ctrl{-1} & \targ & \ctrl{-1} \qw & \ctrl{-1} \qw & \gate{X} & \ctrl{-1} \qw & \ctrl{-1} \qw & \targ & \ctrl{-1} \qw & \ctrl{-1} \qw & \qw \\
}
\end{tabular}

\vspace{1em}
We see that the only classical permutations on three qubits in $\mathcal{CH}$ are Clifford permutations and (up to Clifford permutations) a single Toffoli gate.  

\section{Identifying Permutations in $\mathcal{CH}$}

We have shown that for a generalized semi-Clifford gate to be in the Clifford Hierarchy, its permutations must be in the Clifford Hierarchy. While it is easy to check whether a diagonal gate is an element in $\mathcal{CH}$ (verify that all entries are $2^k$ roots of unity), determining which permutations are in $\mathcal{CH}$ appears to be more challenging. 

Generally, a gate, $U$, is at level $k$ in the Clifford Hierarchy, denoted $\mathcal{CH}_k$, iff $UpU^\dagger = V_{k-1}$ for all $4^n$ Pauli strings $p$ and $V_{k-1}$ is an element of $\mathcal{CH}_{k-1}$. Assuming that membership in $\mathcal{CH}_2$, the Clifford group, can be easily verified, we see that checking if $U$ is in $\mathcal{CH}_3$ requires checking $2n$ Pauli group generators and verifying that each $UpU^\dagger$ is in $\mathcal{CH}_2$. Then, since $Up_1p_2U^\dagger = Up_1U^\dagger Up_2U^\dagger$ and since $\mathcal{CH}_2$ is a group, we can see that if it is true for the generators it must be true for all elements. We can imagine a black box which checks if $UpU^\dagger$ is in $\mathcal{CH}_k$. Then, $2n$ queries must be made to check if $U\in \mathcal{CH}_3$. For levels $k>3$, all $4^n$ strings must be checked (queried) generally. This gives an upper bound on the number of queries to be made of $4^{n(k-3)}2n$. From, Corollary \ref{Xonly}, we see that permutation gates require only $2^{n(k-3)}n$ queries since only Pauli $X$ strings must be checked. This is still exponential. Furthermore, each query generally involves multiplication of $2^n \times 2^n$ matrices. 

Still, there are some classes of permutations which can be shown to be in $\mathcal{CH}$. Here we identify one such class. We will look at circuits of multi-controlled NOTs, denoted $C^n(X)$ for $n\ge 1$. These include CNOT and Toffoli. Note that each gate consists of $n$ controls (in the computational basis) and one target. We will define a time slice of a circuit as some set of $C^n(X)$ gates which can be implemented simultaneously in the circuit. That is, their support is on different wires. We will define the Control-Target ($CT$) mismatch between time slices as the number of controls(targets) of permutation gates in one time slice of a circuit which share the same wire with targets(controls) in another time slice.  

\begin{theorem}\label{thm:2}
If a permutation can be written as a circuit, $U$, such that the CT mismatch is zero between any two time slices, then $U$ is in $\mathcal{CH}$.
\end{theorem}

\begin{proof}
Note that $CT$ mismatch of zero is equivalent to requiring that all time slices commute. Then, since any time slice can be diagonalized by conjugation with local Clifford gates, simply conjugate each target by $H$. And, since the $CT$ mismatch is zero, all time slices can be diagonalized by the same local Clifford gates; $U$ is therefore Clifford-isomorphic to a diagonal gate element in $\mathcal{CH}$ which shows that $U$ must be in $\mathcal{CH}$.
\end{proof}

\begin{lemma}\label{passthru}
A Pauli $X$ string can be `passed through' a permutation gate, $C^n(X)$, without increasing the CT mismatch in the circuit.

We will prove the result for a weight-one Pauli $X$ string, and note that by repeating the result (with other weight-one Pauli $X$ strings), it can be applied to an arbitrary $X$ string. Using the following circuit identity:
\begin{center}
\begin{tabular}{c}
\Qcircuit @C=0.5em @R=0em @!R {
& \gate{X} & \qw & \ctrl{1} & \qw & & & \qw & \ctrl{1} & \gate{X} & \qw \\
& \qw & \qw & \ctrl{1} & \qw & & & \qw & \ctrl{1} & \ctrl{1} & \qw \\
& \qwblank & \qwblank & \qwblank & \qwblank & & & \qwblank & \qwblank & \qwblank & \qwblank \\
& & & \vdots & & \push{\rule{.3em}{0em}=\rule{.3em}{0em}} & & & \vdots & \vdots &  \\
& \qwblank & \qwblank & \ctrlblank{1} & \qwblank & & & \qwblank & \ctrlblank{1} & \ctrlblank{1} & \qwblank \\
& \qw & \qw & \ctrl{1} & \qw & & & \qw & \ctrl{1} & \ctrl{1} & \qw \\
& \qw & \qw & \targ & \qw & & &  \qw & \targ & \targ & \qw 
}
\end{tabular}

\end{center}

we see that the Pauli $X$ string can be passed through any $C^n(X)$ gate without increasing the CT mismatch. The additional $C^{n-1}(X)$ gate has targets and controls at the same locations as the $C^n(X)$ gate and will have the same or lower CT mismatch with all other gates in the circuit. Note that this circuit identity, while derived independently, is similar to the `Pauli Sandwich Trick' of O'Connor and Yoder\cite{OConnor:2021a}.  
\end{lemma}

\begin{corollary}
Since multiplication of $U$ on the left and right by Clifford permutations ($X$s and CNOTs) does not change the level in $\mathcal{CH}$ we can extend Theorem \ref{thm:2} to allow arbitrary Clifford permutations on the left and right of $U$. Also, since Pauli $X$ strings can be `passed through' a permutation circuit without increasing the CT mismatch, we can allow arbitrary $X$ strings between time slices.  
\end{corollary}

\begin{theorem}\label{MM0}
Any permutation circuit with CT mismatch zero is in $\mathcal{CH}$ at the level of the highest-level gate in the circuit. 
\end{theorem}

\begin{proof}
It is well known that $C^n(X)$ is in $\mathcal{CH}$ at level $n+1$ \cite{Zeng2008}. We only need to check Pauli $X$ strings by Lemma \ref{Xonly}. By Lemma \ref{passthru} we see that we can pass Pauli $X$ strings through permutations and the resulting permutations must commute with all other permutations. This allows us to treat each time slice separately, as well as each gate in a time slice if there is more than one. Then, the permutation circuit is in $\mathcal{CH}$ at the level of the highest-level individual gate, as claimed. 
\end{proof}

We close this section with some open problems and a cautionary tale.

\begin{conjecture}
A permutation is in $\mathcal{CH}$ at the third level iff it can be written as a circuit of commuting Toffoli ($C^2(X)$) gates (possibly preceded and followed by Clifford permutations).   
\end{conjecture}

From Thm.~\ref{MM0}, any network of commuting (CT mismatch zero) Toffoli gates is in $\mathcal{CH}_3$. Also, from Thm.~\ref{MM0} we see that if the permutation gates commute, then no $C^n(X)$ gates with $n\ge 3$ are permitted since they belong to $\mathcal{CH}_{k>3}$. What remains to be proven is that any circuit of $C^m(X)$ gates with $m\ge 2$ and CT mismatch greater than zero (non-commuting permutations) is not in $\mathcal{CH}_3$.

\begin{conjecture}
If a permutation $\pi$ is in $\mathcal{CH}_k$, then $\pi^\dagger$ is in $\mathcal{CH}_k$.
\end{conjecture}

We have been careful to use the phrase ``if a permutation can be written as a circuit..." We warn the reader that determining if a permutation can be written in this manner is likely to be difficult, at least in general. Consider the fact that permutations are universal for reversible classical computation, and that the problem of determining whether or not a given reversible circuit is equivalent to the identity is known to be NP-Hard. Given this it is highly unlikely that an efficient algorithm exists to determine the level in $\mathcal{CH}$ of a general permutation. 

\section{Permutation gates with support on $\mathcal{C}$ qubits}\label{PermsOnC}

In this section we look at adding permutation gates which have support on both the Clifford and non-Clifford qubits and the additional restrictions that arise in these cases. These extensions are only relevant for groups such as $\langle \Pi^n, \mathcal{D}_l^{n-m} \times C^m \rangle$ which have a non-Clifford and an (unconstrained) Clifford component. We write all gates in this section in terms of permutations, diagonal gates and Clifford gates, but note that the restrictions below are valid for other groups Clifford-isomorphic to these groups. We also assume that permutations with support on both types of qubits are compatible with arbitrary diagonal gates in $\mathcal{CH}$ on the non-Clifford qubits and arbitrary Clifford gates on the Clifford qubits. This means that some groups with restricted diagonal gates and/or restricted Clifford gates may exist in $\mathcal{CH}$ with a larger set of permutation gates between the Clifford and non-Clifford qubits. Though we have evidence that this does not occur. 

First, note that if all elements in $\Pi^n$ have support exclusively on the $n-m$ qubits with $\mathcal{D}_l^{n-m}$ gates (hereafter referred to as non-Clifford qubits), then these two groups never `interact' and no additional restrictions are needed. Additionally, if $\Pi^n$ has elements that are Clifford permutations which have support exclusively on the $m$ Clifford qubits, no additional restrictions are needed. If either (or both) of these cases occur, we can fully realize the group (or any subgroup of) $(\Pi^{n-m} \ltimes \mathcal{D}_l^{n-m}) \times C^m$ in $\mathcal{CH}$. Since Clifford permutations are part of the Clifford group, they are not written as separate permutations.  

We can immediately rule out a large class of permutations. Permutation gates with target (and potentially, control(s)) on the non-Clifford qubits and with at least one control on the Clifford qubits, will cause the number of non-Clifford qubits to increase. This is true even when the permutation is a $CNOT$ gate as shown below:

\begin{equation}
\begin{array}{ccccc}
\Qcircuit @C=0.5em @R=1.8em {
&\qw & \targ & \gate{T} & \targ & \qw \\
&\qw & \ctrl{-1} & \qw & \ctrl{-1} & \qw \\
}
\end{array}
=
\begin{array}{ccccc}
\Qcircuit @C=0.5em @R=1.3em {
&\qw & \gate{T} & \ctrl{1} & \qw \\
&\qw & \gate{T} & \gate{S^\dagger} & \qw \\
}
\end{array}.
\end{equation}
Here the top qubit is a non-Clifford qubit and the bottom qubit is (was) a Clifford qubit. Even if the non-Clifford diagonal gates are severely restricted to some subgroup such that no non-Clifford gates can `propagate down', Clifford gates such as $\sqrt{X}\triangleq HSH$ can still `propagate up' to the non-Clifford gates. $\sqrt{X}$ is neither a permutation nor a diagonal non-Clifford gate and would create (via multiplication of elements in $G$) an element which is not generalized semi-Clifford. We conclude that permutation gates with any target on the non-Clifford qubits and at least one control on the Clifford qubits are not permitted, given our assumption of arbitrary diagonal gates on the non-Clifford qubits\footnote{As we see from the circuit diagram arbitrary diagonal gates are not necessary and simply a $T$ gate on each qubit is sufficient to restrict this type of permutation. Many other non-Clifford diagonal gate are also sufficient.}. Since this implies that the permutation gates that are (potentially) allowed only have targets on the Clifford qubits, we denote this group on $n$ qubits as $\Pi^n_{\downarrow}$.

Note that any diagonal element on the non-Clifford qubits must commute with any element $\pi\in\Pi^n_{\downarrow}$. That is for any $d\otimes I_m\in D_l^{n-m}\times C^m$ we have $\pi d = d \pi$. Also, for any permutation, $\pi \in \Pi^{n-m}$, with support only on the non-Clifford qubits and $\pi_{\downarrow} \in \Pi^n_{\downarrow}$, we have that $\pi \pi_{\downarrow} \pi^\dagger = \pi_{\downarrow}'$. In other words, elements of $\Pi^{n-m}$ are in the normalizer of $\Pi^n_{\downarrow}$. This is easy to show since the support of any element in $\Pi^n_{\downarrow}$ is equivalent to a diagonal gate on the $n-m$ non-Clifford qubits. We conclude that elements in $\Pi^n_{\downarrow}$ place no additional constraints on elements in $\Pi^{n-m} \ltimes D_l^{n-m}$.

Some elements of $\Pi^n_{\downarrow}$, however, must be restricted when acting on Clifford qubits. We give the following cases which require restrictions on group(s) $\Pi^n_{\downarrow}$ and/or $C^m$:

(1) A permutation gate, $C^t(X) \in \Pi^n_{\downarrow}$, with  $t\ge2$ (a non-Clifford
permutation) which has $supp(C^t(X))=t+1\ge 3$ on the Clifford qubits. Then, since the Clifford group is a maximal finite group and this new gate is non-Clifford, the full 3-qubit Clifford group and this element must generate a infinite group which clearly has elements not in $\mathcal{CH}$. Additional restrictions on one (or both) of the groups must be made. 

(2) A gate, $C^t(X) \in \Pi^n_{\downarrow}$, for $t\ge2$ (a non-Clifford
permutation) has $supp(C^t(X))=t+1\ge 2$ (target and at least one control) on the Clifford qubits. Then, since we assumed that any element of the 2-qubit Clifford group is present on these qubits, we can use products of Hadamards on the Clifford qubits with $C^t(X)$ to implement a permutation not in $\mathcal{CH}$. This group is still finite, but it contains permutations not in $\mathcal{CH}$. See circuit below\footnote{Even if the Clifford group is restricted to $C^1 \times C^1$, a permutation not in $\mathcal{CH}$ can still be implemented (as shown in the circuit). Also, if a subgroup of the 2-qubit Cliffords is used which does not contain any $H$ (Hadamard) gates, gates such as $SWAP$ can also be used to implement permutations not in $\mathcal{CH}$.}:

\begin{equation}\label{No2QCliff}
\begin{tabular}{c}
\Qcircuit @C=0.5em @R=1.3em {
&\qw & \ctrl{1} & \qw & \qw & \ctrl{1} & \qw & \qw \\
&\qw & \ctrl{1} &\qw & \gate{H} & \ctrl{1} & \gate{H} & \qw \\
&\qw & \targ & \qw & \gate{H} & \targ & \gate{H} & \qw 
}
\end{tabular}
=
\begin{tabular}{c}
\Qcircuit @C=0.5em @R=1.7em {
&\qw & \ctrl{1} & \ctrl{1} & \qw \\
&\qw & \ctrl{1} & \targ & \qw \\
&\qw & \targ & \ctrl{-1} & \qw 
}
\end{tabular}
\notin \mathcal{CH} 
\end{equation}
This also requires that one (or both) of the groups is subject to additional restrictions. Since we have made the choice to keep the full Clifford group we restrict the group $\in \Pi^n_{\downarrow}$. If we restrict each element of $\Pi^n_{\downarrow}$ to have support on only one qubit among the Clifford qubits (it must be a target by the discussion above), then we find that any element of this subgroup of $\Pi^n_{\downarrow}$ can be combined with $(\Pi^{n-m} \ltimes \mathcal{D}_l^{n-m}) \times C^m$ without any (additional) constraints. We will denote this subgroup of $\Pi^n_{\downarrow}$ with support equal to one on the Clifford qubits as $\Pi^n_{\downarrow,1}$. Note that elements of $\Pi^n_{\downarrow,1}$ can have support on different (or the same, for that matter) Clifford qubits; however each element has support on only one Clifford qubit. It is easy to show that $\Pi^n_{\downarrow,1}$ is abelian.

Now we can give a (slightly) larger classification of groups in $\mathcal{CH}$. If we demand that permutations in $\Pi^n$ are compatible with arbitrary diagonal gates in $\mathcal{CH}$ on the non-Clifford qubits and arbitrary Clifford gates on the Clifford qubits, then the following are necessary and sufficient conditions. 

Groups of semi-Clifford elements in $\mathcal{CH}$ are Clifford-isomorphic to a subgroup of the following:
\begin{equation*}\label{struct1FULL}
    C_{\Pi}^n \ltimes \mathcal{D}_l^n \mbox{ or } \Pi^n_{C \downarrow,1}( (C_{\Pi}^{n-1} \ltimes \mathcal{D}_l^{n-1}) \times C^1)  \mbox{ or } ... \mbox{ or } \Pi^n_{C \downarrow,1}((C_{\Pi}^1 \ltimes \mathcal{D}_l^{1}) \times C^{n-1}) \mbox{ or } C^n.
\end{equation*}

Here, $\Pi^n_{C \downarrow,1}$ denotes a Clifford element in $\Pi^n_{\downarrow,1}$. These elements are all $CNOT$ gates.

Groups of generalized semi-Clifford elements in $\mathcal{CH}$ are Clifford-isomorphic to a subgroup of the following:
\begin{equation*}\label{struct2FULL}
    \tilde{\Pi}^n \ltimes \mathcal{D}_l^n \mbox{ or } \Pi^n_{\downarrow,1}((\tilde{\Pi}^{n-1} \ltimes \mathcal{D}_l^{n-1}) \times C^1)  \mbox{ or } ... \mbox{ or } \Pi^n_{\downarrow,1}((\tilde{\Pi}^1 \ltimes \mathcal{D}_l^{1}) \times C^{n-1}) \mbox{ or } C^n.
\end{equation*}

We have looked at relaxing the full Clifford group requirement in Equ.~\ref{No2QCliff} to see when a Toffoli gate with target and one control on the (now relaxed) Clifford qubits and at least one control on the non-Clifford gates can be performed. In other words we want to know what Clifford gates can appear in $C_?$ (in the circuit below, see Equ.~\ref{C?}) such that the group generated by these gates and the Toffoli gate are in $\mathcal{CH}$. 

\begin{equation}\label{C?}
\Qcircuit @C=0.5em @R=1.3em {
& \ctrl{1} & \qw & \qw  & \qw \\
& \ctrl{1} & \qw & \multigate{1}{C_{?}} & \qw \\
& \targ    & \qw & \ghost{\mathcal{F}} & \qw \\
}
\end{equation}

In this case we have strong evidence that the $C_?$ must be of the form
\begin{equation*}
\begin{array}{ccc}
\Qcircuit @C=0.5em @R=1.3em {
& \multigate{1}{C_{?}} & \qw \\
& \ghost{\mathcal{F}} & \qw \\
}
\end{array}
= 
\begin{array}{ccc}
\Qcircuit @C=0.5em @R=1.3em {
& \ctrl{1} & \gate{C_{\Pi}} & \qw \\
& \targ & \gate{C^1} & \qw \\
}
\end{array},
\end{equation*}
but this is a subgroup of the $\Pi^2_{\downarrow,1}(C_{\Pi}^1 \ltimes D_l^1)\times C^1)$ where $D_l^1$ is reduced to the trivial subgroup. This is already contained in our classification above. For this classification to be complete we need to show that any non-Clifford permutations acting on the Clifford qubits always restricts the Clifford gates in this manner. (1) and (2) above show that the Clifford gates must be restricted, but we have not proven that they must be restricted to the groups in our classification above. We leave this as an open problem.   

\section[Section Title. Section Subtitle]{Diagonal Gate Groups in $\mathcal{CH}$\\ {\it In which three ways of looking at the same thing are presented.}} \label{sec:diag}

\begin{figure}[!h]
     \centering
     \begin{subfigure}[b]{0.45\textwidth}
         \centering
         \includegraphics[scale=0.75]{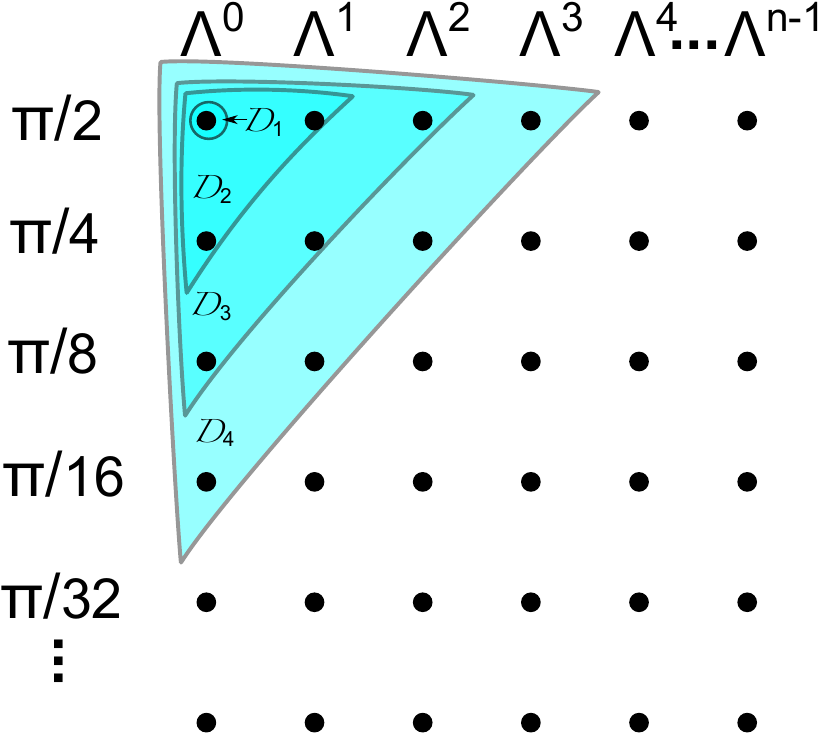}
         \caption{}
         \label{fig:y equals x}
     \end{subfigure}
     \hfill
     \begin{subfigure}[b]{0.45\textwidth}
         \centering
         \includegraphics[scale=0.75]{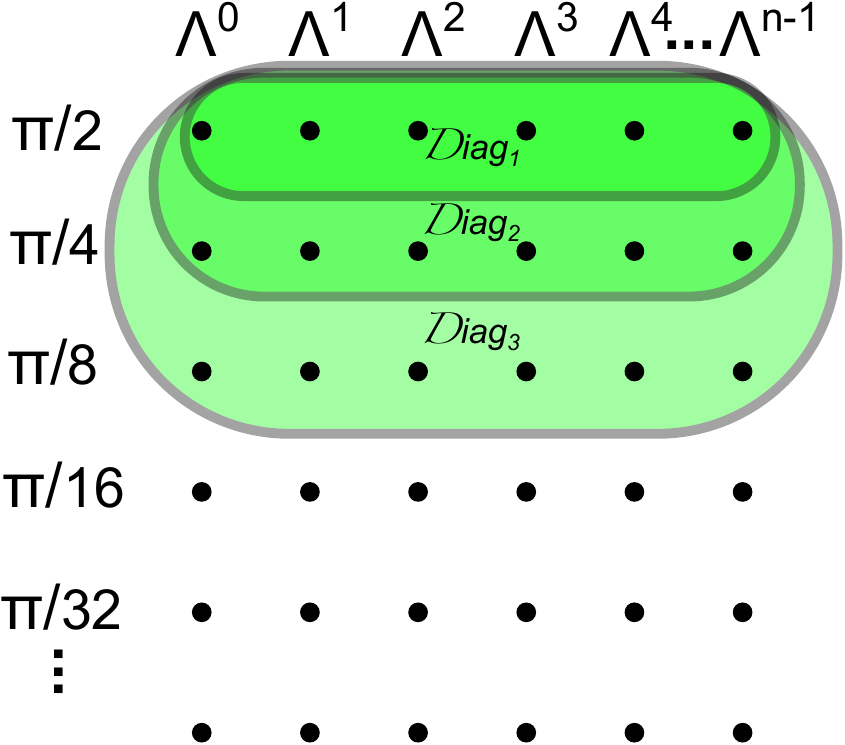}
         \caption{}
         \label{fig:three sin x}
     \end{subfigure}
        \caption{Comparison of diagonal gate groups in $\mathcal{CH}$. Points in figures correspond to controlled-$Z$ rotations by angle $\pi/2^k$. The number of controls is labeled across the top and the rotation angle is on the left. The groups are generated by these multi-controlled rotations. The figure on the left shows the family of groups, $\mathcal{D}_k$, which are the diagonal gates at level $k$ in $\mathcal{CH}$. The figure on the right shows the groups, $Diag_k$, of matrices with entries in the $2^k$th root of unity. Here the gates act on up to $n$ qubits and the generators are on all appropriate sets of qubits. For example, the $\Lambda^2(Z\left[\frac{\pi}{2}\right])$ ($CCZ$ gate) generators are on all $\binom{n}{3}$ qubits. Each group can be generated by the gates (on all appropriate sets of qubits) which are the `lowest' points in each column in the shaded region corresponding to that group.}
        \label{fig:diagGroups}
\end{figure}

Note this section is not necessary for any of the results contained in this paper. I found this way of looking at the diagonal gates in $\mathcal{CH}$ useful and decided to share it.

In Cui \etal \cite{Cui2017DiagonalHierarchy} they classified all the diagonal gates in the Clifford Hierarchy for qubits and for qudits as well. We find that it is helpful to view these gates via three complementary pictures. The first is the $Z$ rotation picture which readily shows what level in $\mathcal{CH}$ an element or group of elements is in. In the $Z$ rotation picture, we write a diagonal element, $d$, as a product of $Z$ rotations as follows:

\begin{equation}\label{diagForm}
    d=\prod_j\exp \left( i\frac{\alpha_j \pi}{2^{k_j}}Z_j \right).
\end{equation}
Here $\alpha_j$ is an integer, $k_j$ is a positive integer, and $Z_j$ is a Pauli $Z$ string. A diagonal gate is in $\mathcal{CH}$ iff it can be written this way and $d$ is in the $k=\max k_j$ level in the Hierarchy. In Cui \etal, they showed that all $n$-qubit $Z$ rotations by a fixed angle $\pi/2^k$ generate the group $\mathcal{D}_k^n$. This is the group of diagonal gates on $n$ qubits in $\mathcal{CH}_k$. Note that all $Z$ rotations by a fixed angle do not necessarily correspond to distinct elements in the group. The group properties are most easily seen in this picture, and we include some results pertaining to the group $\mathcal{D}_k^n$ here. 

\begin{lemma}
For any $d \in \mathcal{D}_k^n$, $d^2 \in \mathcal{D}_{k-1}^n$. 
\end{lemma}
This can be easily seen by looking at the product of $Z$ rotations in Equ.~\ref{diagForm}. A similar observation was made in Campbell and Howard\cite{Campbell2017UnifiedCost}.

\begin{lemma}
For any diagonal gates $d'\notin \mathcal{CH}$ and $d\in\mathcal{CH}$, $d'd \notin \mathcal{CH}$. 
\end{lemma}
To see this, assume that $d'd\in\mathcal{CH}$. Then, $d'd=d_2 \implies d'=d_2 d^\dagger$. And since the diagonal gates on the right-hand side are in $\mathcal{CH}$ via group closure of $\mathcal{D}^n_k$, we have that $d' \in \mathcal{CH}$, a contradiction. 

\begin{lemma}
For any $d\notin \mathcal{CH}$, we have that $\pi d \pi^\dagger \notin \mathcal{CH}$ for any permutation $\pi$. 
\end{lemma}
For any diagonal gate, $d$; conjugation by permutations does not change its eigenvalues. A diagonal gate is in $\mathcal{CH}$ (see Equ.~\ref{diagForm}) iff it has eigenvalues which are $2^k$ roots of unity. A diagonal gate $d\notin \mathcal{CH}$ must therefore have at least one eigenvalue which is not of this form. Since conjugation by permutations does not change this eigenvalue and maps diagonal gates to diagonal gates, we conclude that $\pi d \pi^\dagger \notin \mathcal{CH}$.

\begin{lemma}
For any $d\notin \mathcal{CH}$, we have that $d^2 \notin \mathcal{CH}$. 
\end{lemma}
If $d\notin \mathcal{CH}$, it must have a $Z$ rotation about some angle $\frac{\pi}{r}$ with $r \ne 2^k$ for any $k$. We will write this as $r=2^k p$ where $p\ne 2^k$. Then $d^2$ must have a $Z$ rotation about $\frac{2\pi}{r}$ and since $p$ does not divide 2, we have that $d^2 \notin \mathcal{CH}$. Note that this does not extend to $d^3$ since $\frac{\pi}{3}$ rotations cubed are proportional to Identity which is clearly in $\mathcal{CH}$.  

In the gate picture, we write a diagonal element, $d$, as a product of well-known diagonal gates. Any diagonal gate in $\mathcal{CH}$ on $n$ qubits can be written as some combination of gates:

\begin{equation*}
    \Lambda^{n-1}(Z^{1/2^{k_j}}), \Lambda^{n-2}(Z^{1/2^{k_j}}), \cdots , \Lambda^{1}(Z^{1/2^{k_j}}), Z^{1/2^{k_j}}.
\end{equation*}

With 
\begin{equation*}
Z^{1/2^{k_j}} \triangleq \left [ \begin{array}{lc} 1 & 0 \\ 0 & e^{i\pi/2^{k_j}}\\ \end{array} \right ].    
\end{equation*}

A gate $\Lambda^{m}(Z^{1/2^{l}})$ is in $\mathcal{CH}_{l+m}$. All gates on $n$ qubits with $l+m\le k$ form the group, $\mathcal{D}_k^n$. In the circuit picture all elements are distinct which allows us to count the number of elements in $\mathcal{D}_k^n$.

\begin{equation}
    |\mathcal{D}_k^n| = \prod_{j=0}^{\min (k-1,n-1)}(2^{k-j})^{\binom{n}{j+1}}.
\end{equation}

Finally, in the matrix picture, we have all $2^n \times 2^n$ diagonal matrices with entries in the $2^k$th root of unity. To fix the $U(1)$ gauge, we will always choose the upper-left entry in the matrix to be $1$. We refer to these groups as $Diag_k^n$. We have the following relation between groups: $\mathcal{D}_{k}^n \subseteq Diag_{k}^n \subseteq \mathcal{D}_{k+n}^n \subseteq Diag_{k+n}^n$. In other words, $\mathcal{D}_{k2}^n$ can always be chosen to be large enough to contain all elements of $Diag_{k1}^n$. Therefore, both $Diag^n$ and $\mathcal{D}^n$ contain all diagonal elements in $\mathcal{CH}$. The group $Diag^n_k$ has the property that it is preserved under conjugation by any permutation matrix. 

\begin{lemma}
For any $d \in Diag_k^n$, $d^2 \in Diag_{k-1}^n$. 
\end{lemma}
This follows by noting that all elements in Fig.~\ref{fig:diagGroups} commute and the square of any element in $Diag_k^n$ is an element in $Diag_{k-1}^n$. 

\section{Diagonal Gates in the Qubit Clifford Hierarchy}\label{shortProof}

In an effort to make this paper more self-contained, we provide a simple proof of the diagonal gates in the qubit Clifford Hierarchy. A more general proof classifying all diagonal gates in any prime qudit Clifford Hierarchy is provided in \cite{Cui2017DiagonalHierarchy}. 

First, note that any $2^n \times 2^n$ unitary diagonal matrix, $d$, can be written as a product of Pauli $Z$ rotations 

\begin{equation}\label{equ:genDiag}
    d=\prod_j\exp \left( i\theta_j Z_j \right)
\end{equation}
with $Z_j$ running over all distinct Pauli $Z$ strings and $\theta_j \in [0,2\pi)$. Each Pauli string is unique and there are $2^n$ such strings (one for each diagonal element in a $2^n\times 2^n$ matrix). Note that the `all Identity' string is included in this count, even though it only applies a global (trivial) phase. 

It is well known that an element is in the Clifford group if and only if it can be written as a product of $\pi/4$ Pauli rotations. We can, therefore, write an arbitrary diagonal Clifford element as
\begin{equation*}
    C_d=\prod_j\exp \left( i\frac{\alpha_j\pi}{4} Z_j \right)
\end{equation*}
with $\alpha_j = \{0,\pm 1,\pm 2, \pm 3, \pm 4\}$. Here the $Z_j$ are distinct Pauli $Z$ strings. Note that a product of distinct $\pi/4$ Pauli $Z$ rotations can equal the identity. This is a nuance to consider when counting distinct gates but will not pose a problem here. 

Now for a diagonal gate to be in the third level of $\mathcal{CH}$ it must, under conjugation, take all Pauli strings to a Clifford element. We can express this requirement as  

\begin{equation*}
    d P_i d^\dagger = C_d P_i \hspace{1ex}\forall P_i
\end{equation*}
since
\begin{equation*}
    \begin{cases}
      \exp \left( i\theta_j Z_j \right) P_1 \exp \left( -i\theta_j Z_j \right) = P_1, & \text{if}\ [P_1, Z_j]=0 \\
      \exp \left( i\theta_j Z_j \right) P_1 \exp \left( -i\theta_j Z_j \right) = \exp \left( i2\theta_j Z_j \right) P_1 , & \text{if}\ \{P_1, Z_j\}=0
    \end{cases}
\end{equation*}
for each $Z_j$ noting that $d$ is a product of $Z_j$ rotations. 

Finally, since each $Z_j$ anticommutes with some Pauli string (actually many), we have the following requirement:
\begin{equation*}
   \exp \left( i2\theta_j Z_j \right) = \exp \left( i\frac{\alpha_j \pi}{4} Z_j \right).
\end{equation*}
 We conclude that $\theta_j$ must be a multiple of $\pi/8$ for all $j$. Now that we have classified the diagonal $\mathcal{CH}_3$ gates, we can proceed to find the diagonal gates in the $4$th level of $\mathcal{CH}$ which, under conjugation, take all Pauli strings to a diagonal gate in $\mathcal{CH}_3$ times a Pauli string. A similar proof shows that the diagonal gates in $\mathcal{CH}_4$ must be products of diagonal rotations with $\theta_j \propto \pi/16$, and more generally, that a diagonal gate element in $\mathcal{CH}_k$ must be a product of diagonal rotations with $\theta_j \propto \pi/2^k$. We have sketched a proof outline here, but note that a straightforward proof by induction is required to make this proof outline rigorous.
 
 \section{Some properties of Generalized Symmetric Groups pertaining to this work}\label{sec:GSG}
 
 Generalized symmetric groups can be expressed as semi-direct products of the symmetric group on $N$ elements, $S_N$, and $N$ copies of a cyclic group $\mathbb{Z}_M$. $\mathbb{Z}_l^N$ has a unitary representation as $N\times N$ diagonal matrices with entries in $\mathbb{Z}_m$ expressed as complex $M$ roots of unity ($e^{i 2\pi/M}$), and $S_N$ has a standard representation as $N\times N$ permutation matrices. We denote these generalized symmetric groups as $S(M,N)$. For finite $M, N$, these groups are finite and have size $|S(M,N)|=M^NN!$.
 
 Here all the representations correspond to unitary gates on $n$ qubit which are sized $2^n \times 2^n$ matrices and $N=2^n$. Additionally, due to the constraints on diagonal gates in $\mathcal{CH}$, we always have $M=2^k$, expressing the fact that diagonal gates here must have entries which are some $2^k$th root of unity. These groups $S(2^k,2^n)$ are of order $(2^k)^{2^n}(2^n!)=2^{k2^n}(2^n!)$. 
 
 Elements of $S(2^k,2^n)$ can always be expressed as a product of a permutation, $\pi$, with a diagonal matrix, $d$. We can write this as $g=(\pi,d)=\pi d$. Then, it is easy to show that for any two elements we have $g_1 g_2 = (\pi_1,d_1)(\pi_2,d_2)= \pi_1d_1\pi_2d_2 = \pi_1\pi_2 (\pi_2^{-1} d_1 \pi_2 d_2) = (\pi_1\pi_2, \pi_2^{-1} d_1 \pi_2 d_2) = g_3$.
 
 We are interested in subgroups of $S(2^k,2^n)$ since these all correspond to groups in $\mathcal{CH}$. Generally finding all such subgroups seems difficult, but we note that for all subgroups of $S(2^k,2^n)$, the `permutation part' of each element (the $\pi$'s in each $g=(\pi, d)$) must form a subgroup of the permutation group. 
 
 We sidestep much of the difficultly of finding all such subgroups by placing constraints on generators, such that any set of generators satisfying the constraints listed in this paper will generate a valid subgroup of $S(2^k,2^n)$. Furthermore, since this is a finite group and, therefore, finitely generated, we can express all such subgroups this way. 
 
 By examination of the product of an element $g$ with itself we see that if $g=(\pi,d)$ and $\pi^m = I$, then $g^m$ is a purely diagonal entry. When these diagonal entries are non-trivial, we have (non-trivial) purely diagonal subgroups within the subgroup of $S(2^k,2^n)$.

\bibliographystyle{quantum}


\end{document}